\newif\ifreport\reporttrue
\documentclass[twoside]{IEEEtran}
\usepackage{amsmath,amssymb,amsfonts,amsthm}
\usepackage{subcaption}
\usepackage{color}
\usepackage{algorithm}
\DeclareMathOperator*{\argmax}{arg\,max}
\DeclareMathOperator*{\argmin}{arg\,min}
\usepackage{algpseudocode}

\newtheorem{theorem}{Theorem}

\newtheorem{definition}{Definition}
\newtheorem{lemma}{Lemma}

\def\blue{\color{blue}}
\def\red{\color{red}}

\usepackage{tcolorbox}
\usepackage{lipsum}
\usepackage{hyperref}
\usepackage[noadjust]{cite}
\colorlet{red}{black}

\begin{document}

\title{Learning
 and Communications Co-Design for Remote Inference Systems: Feature Length Selection and Transmission Scheduling}
        
\author{Md Kamran Chowdhury Shisher,~\IEEEmembership{Student~Member,~IEEE,} Bo~Ji,~\IEEEmembership{Senior~Member,~IEEE,}
I-Hong~Hou,~\IEEEmembership{Senior~Member,~IEEE,}
        Yin~Sun,~\IEEEmembership{Senior~Member,~IEEE}
\IEEEcompsocitemizethanks{\IEEEcompsocthanksitem M.K.C. Shisher and Y. Sun are with the Department of Electrical and Computer Engineering, Auburn University, Auburn, AL 36849 USA (e-mail: mzs0153@auburn.edu, yzs0078@auburn.edu).  B. Ji is with the Department of Computer Science, Virginia Tech, Blacksburg, VA 24061 USA (e-mail: boji@vt.edu). I.-H. Hou is with the Department of Electrical and Computer Engineering, Texas A\&M University, College Station, TX 77843 USA (e-mail: ihou@tamu.edu).}
\thanks{The work of M.K.C. Shisher and Y. Sun was supported in part by the NSF grant CNS-2239677 and the ARO grant W911NF-21-1-0244. The work of B. Ji was supported in part by the NSF under Grants CNS-2112694 and CNS-2106427. The work of I.-H. Hou was supported in part by NSF under Award Number ECCS-2127721, in part by the U.S. Army Research Laboratory and the U.S. Army Research Office under Grant Number W911NF-22-1-0151, and in part by Office of Naval Research under Contract N00014-21-1-2385.}}

\newcommand{\ignore}[1]{{}}
\pagestyle{plain}
\def\blue{\color{blue}}
\maketitle

\begin{abstract} 
In this paper, we consider a remote inference system, where a neural network is used to infer a time-varying target (e.g., robot movement), based on features (e.g., video clips) that are progressively received from a sensing node (e.g., a camera). Each feature is a temporal sequence of sensory data. The inference error is determined by (i) the timeliness and (ii) the sequence length of the feature, where we use Age of Information (AoI) as a metric for timeliness. While a longer feature can typically provide better inference performance, it often requires more channel resources for sending the feature. To minimize the time-averaged inference error, we study a learning and communication co-design problem that jointly optimizes feature length selection and transmission scheduling. When there is a single sensor-predictor pair and a single channel, we develop low-complexity optimal co-designs for both the cases of time-invariant and time-variant feature length. When there are multiple sensor-predictor pairs and multiple channels, the co-design problem becomes a restless multi-arm multi-action bandit problem that is PSPACE-hard. For this setting, we design a low-complexity algorithm to solve the problem. Trace-driven evaluations demonstrate the potential of these co-designs to reduce inference error by up to $10000$ times.
\end{abstract}


%{\blue Example of blue mark.} {\red Example of red mark.}
\section{Introduction}
The advancement of communication technologies and artificial intelligence has engendered the demand for remote inference in various applications, such as autonomous vehicles, health monitoring, industrial control systems, and robotic systems \cite{6GAI, bojarski2016end, rahmatizadeh2018vision, she2020deep}. For instance, accurate prediction of the robotic state during remote robotic surgery is time-critical. The remote inference problem can be tackled by using a neural network that is trained to predict a time-varying target (e.g. robot movement) based on features (e.g., video clips) sent from a remote sensing node (e.g. a camera). Each feature is a temporal sequence of the sensory output and the length of the temporal sequence is called \emph{feature length}.

Due to data processing time, transmission errors, and transmission delay, the features delivered to the neural predictor may not be fresh, which can significantly affect the inference accuracy. To measure the freshness of the delivered features, we use the \emph{age of information (AoI)} metric, which was first introduced in \cite{kaul2012real}. Let $U(t)$ be the generation time of the most recently delivered feature sequence. Then, AoI is the time difference between the generation time $U(t)$ and the current time $t$, denoted by $\Delta(t):=t-U(t)$. Recent studies \cite{shisher2021age, ShisherMobihoc} have shown that the inference error is a function of AoI for a given feature length, but this function is \emph{not necessarily monotonic}. Moreover, simulation results in \cite{shisher2021age} suggest that AoI-aware remote inference, wherein both the feature and its AoI are fed to the neural network, can achieve superior performance than AoI-agnostic remote inference that omits the provision of AoI to the neural network. Hence, the AoI $\Delta(t)$ can provide useful information for reducing the inference error. 

Additionally, the performance of remote inference depends on the sequence length of the feature. Longer feature sequences can carry more information about the target, resulting in the reduction of inference errors. Though a longer feature can provide better training and inference  performance, it often requires more communication resources. For example, a longer feature may require a longer transmission time and may end up being stale when delivered, thus resulting in worse inference performance. Hence, it is necessary to study a learning and communications co-design problem that jointly controls the timeliness and the length of the feature sequences. The contributions of this paper are summarized as follows: 
\begin{itemize}
\item In \cite{ShisherMobihoc}, it was demonstrated that the inference error is a function of the AoI, whereas the function is not necessarily monotonic. The current paper further investigates the impact of feature length on inference error. Our information-theoretic and experimental analysis show that the inference error is a non-increasing function of the feature length (See Figs. \ref{fig:Trainingcsi}(a)-\ref{fig:Trainingcartpoleangle}(a), and Lemma \ref{lemma2}).

\item We propose a novel learning and communications co-design framework (see Sec. \ref{SystemAndPolicy}). In this framework, we adopted the ``selection-from-buffer” model proposed in \cite{ShisherMobihoc}, which is more general than the popular “generate-at-will” model that was proposed in \cite{yates2015lazy} and named in \cite{sun2017update}. In addition, we consider both time-invariant and time-variant feature length. Earlier studies, for example \cite{ShisherMobihoc, tripathi2021computation}, did not consider time-variant feature length. 

\item For a single sensor-predictor pair and a single channel, this paper jointly optimizes feature length selection and transmission scheduling to minimize the time-averaged inference error. This joint optimization is formulated as an infinite time-horizon average-cost \emph{semi-Markov decision process} (SMDP). Such problems often lack analytical solutions or closed-form expressions. Nevertheless, we are able to derive a closed-form expression for an optimal scheduling policy in the case of time-invariant feature length (Theorem \ref{theorem1}). The optimal scheduling time strategy is a \emph{threshold-based policy}. Our threshold-based scheduling approach differs significantly from previous threshold-based policies in e.g., \cite{ShisherMobihoc, SunNonlinear2019, orneeTON2021, klugel2019aoi, Tripathi2019}, because our threshold function depends on both the AoI value and the feature length, while prior threshold functions rely solely on the AoI value. In addition, our threshold function is \emph{not necessarily monotonic} with AoI. This is a significant difference with prior studies \cite{SunNonlinear2019, orneeTON2021, klugel2019aoi, Tripathi2019}.   

\item We provide an optimal policy for the case of time-variant feature length. Specifically, Theorem \ref{theorem2} presents the Bellman equation for the average-cost SMDP with time-variant feature length. The Bellman equation can be solved by applying either relative value iteration or policy iteration algorithms \cite[Sec. 11.4.4]{puterman2014markov}. Given the complexity associated with converting the average-cost SMDP into a Markov Decision Process (MDP) suitable for relative value iteration, we opt for the alternative: using the \emph{policy iteration} algorithm to solve our average-cost SMDP. By leveraging specific structural properties of the SMDP, we can simplify the policy iteration algorithm to reduce its computational complexity. The simplified  policy iteration algorithm is outlined in Algorithm \ref{alg:PolicyEvaluation} and Algorithm \ref{alg:PolicyIteration}.

\item Furthermore, we investigate the learning and communications co-design problem for multiple sensor-predictor pairs and multiple channels. This problem is a restless multi-armed, multi-action bandit problem that is known to be PSPACE-hard \cite{papadimitriou1994complexity}. Moreover, proving indexability condition relating to Whittle index policy \cite{whittle1988restless} for our problem is fundamentally difficult. To this end, we propose a new scheduling policy named ``Net Gain Maximization" that does not need to satisfy the indexability condition (Algorithm \ref{alg:multischeduling}).

\item Numerical evaluations demonstrate that our policies for the single source case can achieve up to $10000$ times performance gain compared to periodic updating and zero-wait policy (see Figs. \ref{fig:numericalsingle1}-\ref{fig:numericalsingle2}). Furthermore, our proposed multiple source policy outperforms the maximum age-first policy (see Fig. \ref{fig:numericalmultiple11}) and is close to a lower bound (see Fig. \ref{fig:numericalmultiple1}).
\end{itemize}

\subsection{Related Works}
%The age of information (AoI) has gained significant attention as a key metric in the optimization of low-latency networked intelligent systems. As surveyed in \cite{yates2021age}, several studies have investigated sampling and scheduling policies for minimizing linear and nonlinear functions of AoI \cite{ShisherMobihoc, sun2017update, Kadota2018, Kadota2019, SunNonlinear2019, Tripathi2019, klugel2019aoi, sombabu2020whittle, chen2021whittle, hsu2018age, sun2019closed, bedewy2021optimal, bedewy2020optimizing, pan2020minimizing}. {In most previous works \cite{sun2017update, Kadota2018, Kadota2019, SunNonlinear2019, Tripathi2019, klugel2019aoi, sombabu2020whittle, chen2021whittle, hsu2018age, sun2019closed, bedewy2021optimal, bedewy2020optimizing, pan2020minimizing}, monotonic AoI penalty functions are considered. However, in a recent study \cite{ShisherMobihoc}, it is demonstrated that the monotonic assumption is not always true for remote inference. In contrast, the inference error is a function of AoI, but the function is not necessarily monotonic. The present paper further investigates the impact of feature length on the inference error by conducting machine learning experiments and by utilizing information-theoretic measures: $L$-conditional entropy and $L$-conditional cross entropy \cite{ShisherMobihoc}, which are more general than the information theoretic metrics used in prior AoI studies \cite{soleymani2016optimal,SunNonlinear2019, chen2021uncertainty}.

The age of information (AoI) has emerged as a popular metric for analyzing and optimizing communication networks \cite{Kadota2018, Kadota2019}, control systems \cite{klugel2019aoi, soleymani2019stochastic}, remote estimation \cite{SunTIT2020, orneeTON2021}, and remote inference \cite{shisher2021age, ShisherMobihoc}. 
As surveyed in \cite{yates2021age}, several studies have investigated sampling and scheduling policies for minimizing linear and nonlinear functions of AoI \cite{ShisherMobihoc, sun2017update, Kadota2018, Kadota2019, SunNonlinear2019, Tripathi2019, klugel2019aoi, sombabu2020whittle, chen2021whittle, hsu2018age, sun2019closed, bedewy2021optimal, bedewy2020optimizing, pan2020minimizing}. In most previous works \cite{sun2017update, Kadota2018, Kadota2019, SunNonlinear2019, Tripathi2019, klugel2019aoi, sombabu2020whittle, chen2021whittle, hsu2018age, sun2019closed, bedewy2021optimal, bedewy2020optimizing, pan2020minimizing}, monotonic AoI penalty functions are considered. However, in a recent study \cite{ShisherMobihoc}, it is demonstrated that the monotonic assumption is not always true for remote inference. In contrast, the inference error is a function of AoI, but the function is not necessarily monotonic. The present paper further investigates the impact of feature length on the inference error and jointly optimizes AoI and feature length.

In recent years, researchers have increasingly employed information-theoretic metrics to evaluate information freshness \cite{ soleymani2016optimal, SunSPAWC2018, SunNonlinear2019, shisher2021age, ShisherMobihoc, chen2021uncertainty, chen2023index, wang2022framework}. In \cite{soleymani2016optimal,SunSPAWC2018, SunNonlinear2019}, the authors utilized Shannon's mutual information to quantify the amount of information carried by received data messages about the current source value, and used Shannon's conditional entropy to measure the uncertainty about the current source value after receiving these messages. These metrics were demonstrated to be monotonic functions of the AoI when the source follows a time-homogeneous Markov chain \cite{SunSPAWC2018, SunNonlinear2019}. Built upon these findings, the authors of \cite{wang2022framework} extended this framework to include hidden Markov model. Furthermore, a Shannon’s conditional entropy term $H_{\mathrm{Shannon}}(Y_t|X_{t-\Delta(t)}=x)$ was used in \cite{chen2021uncertainty, chen2023index} to quantify information uncertainty. However, a gap still existed between these information-theoretic metrics and the performance of real-time applications such as remote estimation or remote inference. In our recent works \cite{shisher2021age, ShisherMobihoc, Shisher2023Timely} and the present paper, we have bridged this gap by using a generalized conditional entropy associated with a loss function $L$, called $L$-conditional entropy, to measure (or approximate) training and inference errors in remote inference, as well as the estimation error in remote estimation. For example, when the loss function $L (y, \hat y)$ is chosen as a quadratic function $|| y - \hat y||_2^2$, the $L$-conditional entropy $H_L(Y_t|X_{t-\Delta(t)})=\min_{\phi}\mathbb E[(Y_t-\phi(X_{t-\Delta(t)}))^2]$ is exactly the minimum mean squared estimation error in remote estimation. This approach allows us to analyze how the AoI $\Delta(t)$ affects inference and estimation errors directly, instead of relying on information-theoretic metrics as intermediaries for assessing application performance. It is worth noting that Shannon’s conditional entropy is a special case of $L$-conditional entropy, corresponding to the inference and estimation errors for softmax regression and maximum likelihood estimation, as discussed in Section \ref{SystemAndPolicy}.

\begin{figure*}[ht]
\centering
\includegraphics[width=0.6\textwidth]{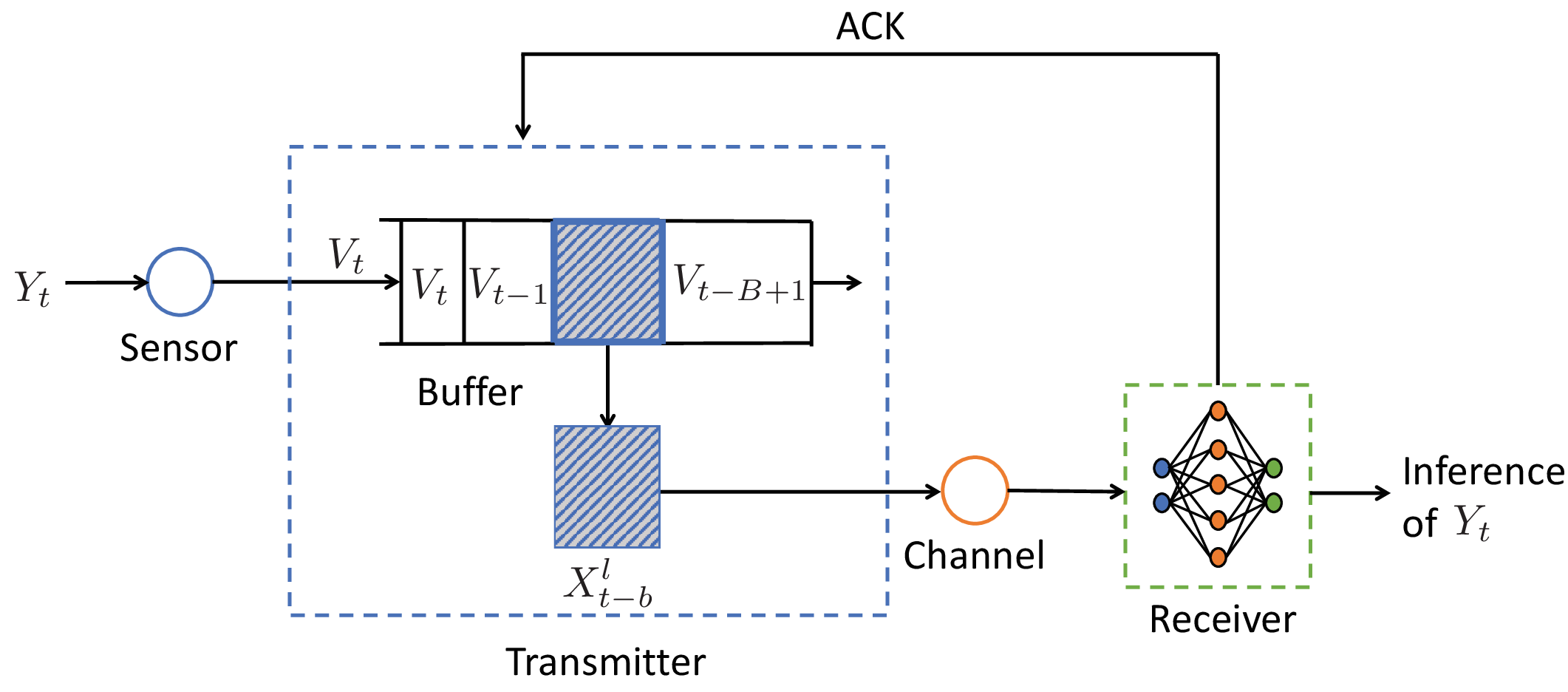}
\caption{\small  A remote inference system, where $X_{t-b}^l:= (V_{t-b}, V_{t-b-1},\ldots, V_{t-b-l+1})$ is a feature with sequence length $l$.  \label{fig:scheduling}
}
\vspace{-3mm}
\end{figure*}

{The optimization of linear and non-linear functions of AoI for multiple source scheduling can be formulated as a restless multi-armed bandit problem \cite{Tripathi2019, ShisherMobihoc, xiong2022index, zou2021minimizing, chen2021scheduling}. Whittle, in his seminal work  \cite{whittle1988restless}, proposed an index-based policy to address restless multi-armed bandit (RMAB) problems with binary actions. Our multiple source scheduling problem is a RMAB problem with multiple actions. An extension of the Whittle index policy for multiple actions was provided in \cite{hodge2015asymptotic}, but it requires to satisfy a complicated \emph{indexability} condition. In \cite{tripathi2021computation}, the authors considered joint feature length selection and transmission scheduling, where the penalty function was assumed to be non-decreasing in the AoI, the feature length is time-invariant, and there is only one communication channel. Under these assumptions, \cite{tripathi2021computation} established the indexability condition and developed a Whittle Index policy. Compared to \cite{tripathi2021computation}, our work could handle both monotonic and non-monotonic AoI penalty functions, both time-invariant and time-variant feature lengths, and both one and multiple communication channels. 

Because of (i) the time-variant feature length and non-monotonic AoI penalty function and (ii) the fact that there exist multiple transmission actions, we could not utilize the Whittle index theory to establish indexability for our multiple source scheduling problem. To address this challenge, we propose a new ``Net Gain Maximization" algorithm (Algorithm \ref{alg:multischeduling}) for multi-source feature length selection and transmission scheduling, which does not require  indexability. During the revision of this paper, we found a related study \cite{chen2023index}, where the authors introduced a similar gain index-based policy for a RMAB problem with two actions: to send or not to send. The ``Net Gain Maximization" algorithm that we propose is more general than the gain index-based policy in \cite{chen2023index} due to its capacity to accommodate more than two actions in the RMAB.}

\section{System Model and Scheduling Policy}\label{SystemAndPolicy}

We consider a remote inference system composed of a sensor, a transmitter, and a receiver, as illustrated in Fig. \ref{fig:scheduling}. The sensor observes a time-varying target $Y_t \in \mathcal Y$ and feeds its measurement $V_t \in \mathcal V$ to the transmitter. The transmitter generates features from the sensory outputs and progressively  transmits the features to the receiver through a communication channel. Within the receiver, a neural network infers the time-varying target based on the received features. 

\subsection{System Model}
The system is time-slotted and starts to operate at time slot $t = 0$. At every time slot $t$, the transmitter appends the sensory output $V_t \in \mathcal V$ to a buffer that stores the $B$ most recent sensory outputs $(V_t, V_{t-1}, \ldots, V_{t-B+1});$ meanwhile, the oldest output $V_{t-B}$ is removed from the buffer. We assume that the buffer is full initially, containing $B$ signal values $(V_0, V_1, \ldots, V_{-B+1})$ at time $t=0$. This ensures that the buffer remains consistently full at any time $t$.\footnote{This assumption does not introduce any loss of generality. If the buffer is no full at time $t=0$, it would not affect our results.} The transmitter progressively generates a feature $X_{t-b}^l$, where each feature $X_{t-b}^l:= (V_{t-b},\ldots,V_{t-b-l+1})\in \mathcal V^l$ is a temporal sequence of sensory outputs taken from the buffer such that $\mathcal V^l$ is the set of all $l$-tuples that take values from $\mathcal V$, $1\leq l\leq B$, and $0\leq b\leq B-l$. For ease of presentation, the temporal sequence length $l$ of feature $X_{t-b}^l$ is called \emph{feature length} and the starting position $b$ of feature $X_{t-b}^l$ in the buffer is called \emph{feature position}. If the channel is idle in time slot $t$, the transmitter can submit the feature $X_{t-b}^l$ to the channel. Due to communication delays and channel errors, the feature is not instantly received. The most recently received feature is denoted as $X^l_{t-\delta} = (V_{t-\delta}, V_{t-\delta-1}, \ldots, V_{t-\delta-l+1})$, where the latest observation $V_{t-\delta}$ in feature $X^l_{t-\delta}$ is generated $\delta$ time slots ago. We call $\delta$ the \emph{age of information (AoI)} which represents the difference between the time stamps of the target $Y_t$ and the latest observation $V_{t-\delta}$ in feature $X^l_{t-\delta}$. 

\begin{figure*}
  \centering
\begin{subfigure}[b]{0.55\columnwidth}
\includegraphics[width=\linewidth]{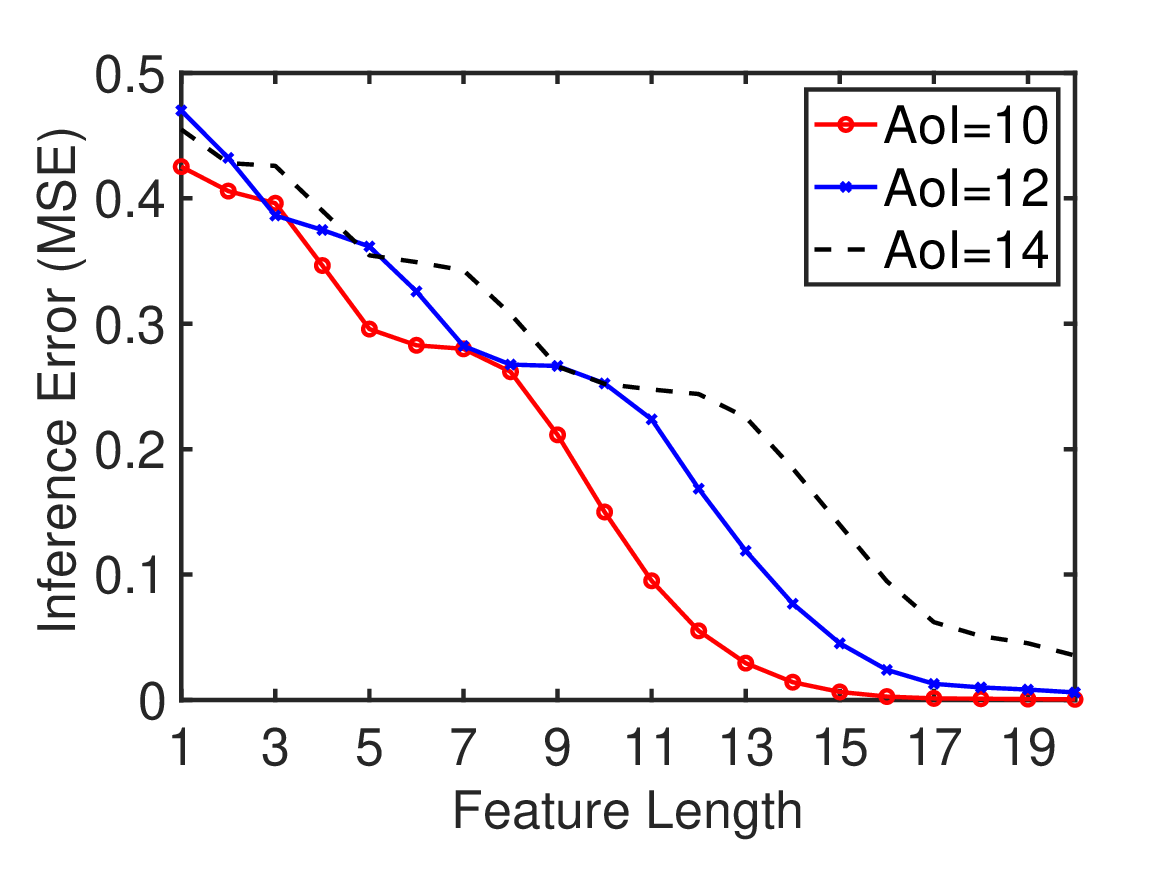}
  \subcaption{\small Inference error vs. Feature length}
\end{subfigure}
 \hspace{20mm}
\begin{subfigure}[b]{0.55\columnwidth}
\includegraphics[width=\linewidth]
{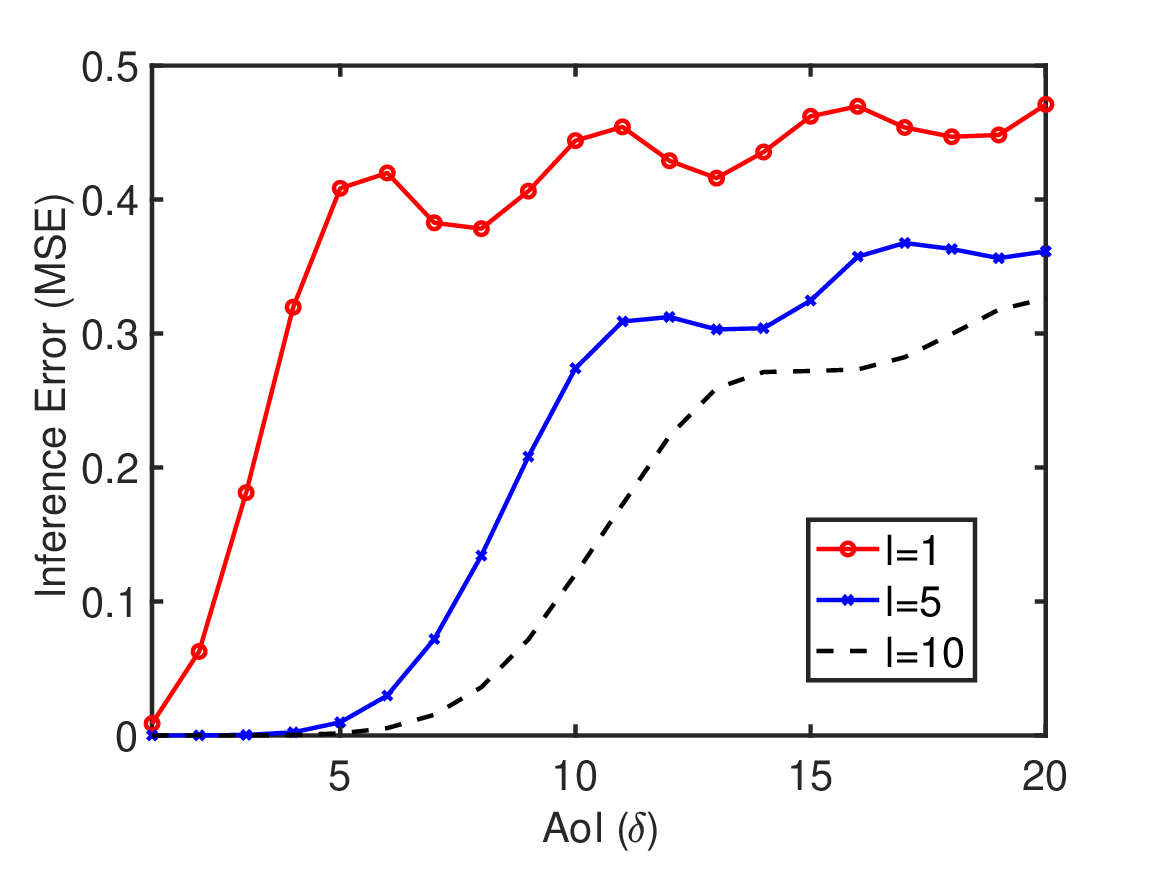}
  \subcaption{\small Inference error vs. AoI}
\end{subfigure}
%
%\begin{subfigure}{0.3\textwidth}
%\includegraphics[width=\textwidth]{Cartpole_Angle.eps}
%  \subcaption{CartPole Angle Prediction}
%\end{subfigure}
%%
\caption{\small Performance of wireless channel state information prediction: (a) Inference error Vs. Feature length and (b) Inference error Vs. AoI.}
\label{fig:Trainingcsi}
\end{figure*} 

{The receiver consists of $B$ trained neural networks, each associated with a specific feature length $l=1, 2, \ldots, B$. The neural network associated with feature length $l$ takes the AoI $\delta \in \mathbb Z^{+}$ and the feature $X^l_{t-\delta} \in \mathcal V^l$ as inputs and generates an output $a = \phi_{l}(\delta, X^l_{t-\delta}) \in \mathcal A$, where the neural network is represented by the function $\phi_l: \mathbb Z^{+} \times \mathcal V^l \mapsto \mathcal A$. The performance of the neural network is measured by a loss function $L: \mathcal Y \times \mathcal A \mapsto \mathbb R$, where $L(y, a)$ indicates the incurred loss if the output $a \in \mathcal A$ is used for inference when $Y_t = y$. The loss function $L$ is determined by the purpose of the application. For example, in softmax regression (i.e., neural network based maximum likelihood classification), the output $a=Q_Y$ is a distribution of $Y_t$ and the loss function $L_{\text{log}}(y, Q_Y)=-\text{log}~Q_Y(y)$ is the negative log-likelihood of the value $Y_t=y$. In neural network based mean-squared estimation, a quadratic loss function $L_2(\mathbf y, \mathbf{\hat y})=\|\mathbf y-\mathbf{\hat y}\|^2_2$ is used, where the action $a=\mathbf{\hat y}$ is an estimate of the target value $Y_t=\mathbf y$ and $\| \mathbf y \|_2$ is the euclidean norm of the vector $\mathbf y$.

\subsection{Inference Error}\label{Experiments} 
We assume that $\{(Y_t , X^l_t ), t \in \mathbb Z\}$ is a stationary process for every $l$. Given AoI $\delta$ and feature length $l$, the expected inference error is a function of $\delta$ and $l$, given by
\begin{align}\label{instantaneous_err1} 
\mathrm{err}_{\mathrm{inference}}(\delta, l):=\mathbb E_{Y, X^l \sim P_{Y_t, X^{l}_{t-\delta}}}\bigg[L\bigg(Y,\phi_{l}\bigg(\delta, X^l\bigg)\bigg)\bigg],
\end{align}
where $P_{Y_t, X^{l}_{t-\delta}}$ is the joint distribution of the label $Y_t$ and feature $X^{l}_{t-\delta}$ during online inference and the function $\phi_{l}$ represents any trained neural network that maps from $\mathbb Z^{+} \times \mathcal V^l$ to $\mathcal A$. The inference error $\mathrm{err}_{\mathrm{inference}}(\delta, l)$ can be evaluated through machine learning experiments.}

In this paper, we conduct two experiments: (i) wireless channel state information (CSI) prediction and (ii) actuator states prediction in the OpenAI CartPole-v1 task \cite{brockman2016openai}. Detailed information regarding the experimental setup for both experiments can be found in Appendix A of the supplementary material. The code for these experiments is available in GitHub repositories\footnote{\url{https://github.com/Kamran0153/Channel-State-Information-Prediction}}\footnote{\url{https://github.com/Kamran0153/Impact-of-Data-Freshness-in-Learning}}.

The experimental results, presented in Figs. \ref{fig:Trainingcsi}(a)-\ref{fig:Trainingcartpoleangle}(a), demonstrate that the inference error decreases with respect to feature length. Moreover, Figs. \ref{fig:Trainingcsi}(b)-\ref{fig:Trainingcartpoleangle}(b) illustrate that the inference error is not necessarily a monotonic function of AoI. 
These findings align with machine learning experiments conducted in \cite{shisher2021age, ShisherMobihoc, Shisher2023Timely}. Collectively, the results from this paper and those in \cite{shisher2021age, ShisherMobihoc, Shisher2023Timely} indicate that longer feature lengths can enhance inference accuracy and fresher features are not always better than stale features in remote inference.

\subsection{Feature Length Selection and Transmission Scheduling Policy}\label{TransmissionPolicy} 
Because (i) fresh feature is not always better than stale feature and (ii) longer feature can improve inference error, we adopted ``selection-from-buffer" model, which is recently proposed in \cite{ShisherMobihoc}. In contrast to the ``generate-at-will" model \cite{sun2017update, yates2015lazy}, where the transmitter can only select the most recent sensory output $V_t$, the ``selection-from-buffer" model offers greater flexibility by allowing the transmitter to pick multiple sensory outputs (which can be stale or fresh). In other words, ``selection-from-buffer" model allows the transmitter to choose feature position $b$ and feature length $l$ under the constraints $1 \leq l \leq B-1$ and $0 \leq b \leq B-l$. Feature length selection represents a trade-off between learning and communications: A longer feature can provide better learning performance (see Figs. \ref{fig:Trainingcsi}-\ref{fig:Trainingcartpoleangle}), whereas it requires more channel resources (e.g., more time slots or more frequency resources) for sending the feature. This motivated us to study a learning-communication co-design problem that jointly optimizes the feature length, feature position, and transmission scheduling.

The feature length and feature position may vary across the features sent over time. Feature transmissions over the channel are non-preemptive: the channel must finish sending the current feature, before becoming available to transmit the next feature. Suppose that the $i$-th feature $X_{S_i-b_i}^{l_i} = (V_{S_i-b_i}, V_{S_i-b_i-1},\ldots, V_{S_i-b_i-l_i+1})$ is submitted to the channel at time slot $t=S_i$, where $l_i$ is its feature length and $b_i$ is its feature position such that $1\leq l_i\leq B$ and $0\leq b_i\leq B-l_i$. It takes $T_i(l_i)\geq 1$ time slots to send the $i$-th feature over the channel. The $i$-th feature is delivered to the receiver at time slot $D_i=S_i + T_i(l_i)$, where $S_i < D_i \leq S_{i+1}$. The feature transmission time $T_i(l_i)$ depends on the feature length $l_i$. Due to time-varying channel conditions, we assume that, given feature length $l_i=l$, the $T_i(l)$'s are \emph{i.i.d.} random variables, with a finite mean $1\leq \mathbb E[T_i(l)] < \infty$. Once a feature is delivered, an acknowledgment (ACK) is sent back to the transmitter, notifying that the channel has become idle.

In time slot $t$, the $i(t)$-th feature $X_{S_{i(t)}-b_{i(t)}}^{l_{i(t)}}$ is the most recently received feature, where $i(t)=\max_i\{D_i\leq t\}$. The receiver feeds the feature $X_{S_{i(t)}-b_{i(t)}}^{l_{i(t)}}$ to the neural network to infer $Y_t$. We define \emph{age of information (AoI)} $\Delta(t)$ is defined as the difference between the time-stamp of the freshest sensory output $V_{S_{i(t)}-b_{i(t)}}$ in feature $X_{S_{i(t)}-b_{i(t)}}^{l_{i(t)}}$ and the current time $t$, i.e.,  
\begin{align}\label{Def_AoI}
\Delta(t):=t-\max_i \{S_i-b_i: D_i \leq t \}.
\end{align}
Because $D_i < D_{i+1}$, it holds that  
\begin{align}\label{Def_AoI1}
\Delta(t)=t-S_i+b_i,~\text{if}~ D_i \leq t < D_{i+1}.
\end{align}
The initial state of the system is assumed to be $S_0=0, l_0=1, b_0=0, D_0=T_0(l_0)$, and $\Delta(0)$ is a finite constant. 

\begin{figure*}[ht]
  \centering
\begin{subfigure}[b]{0.6\columnwidth}
\includegraphics[width=\linewidth]{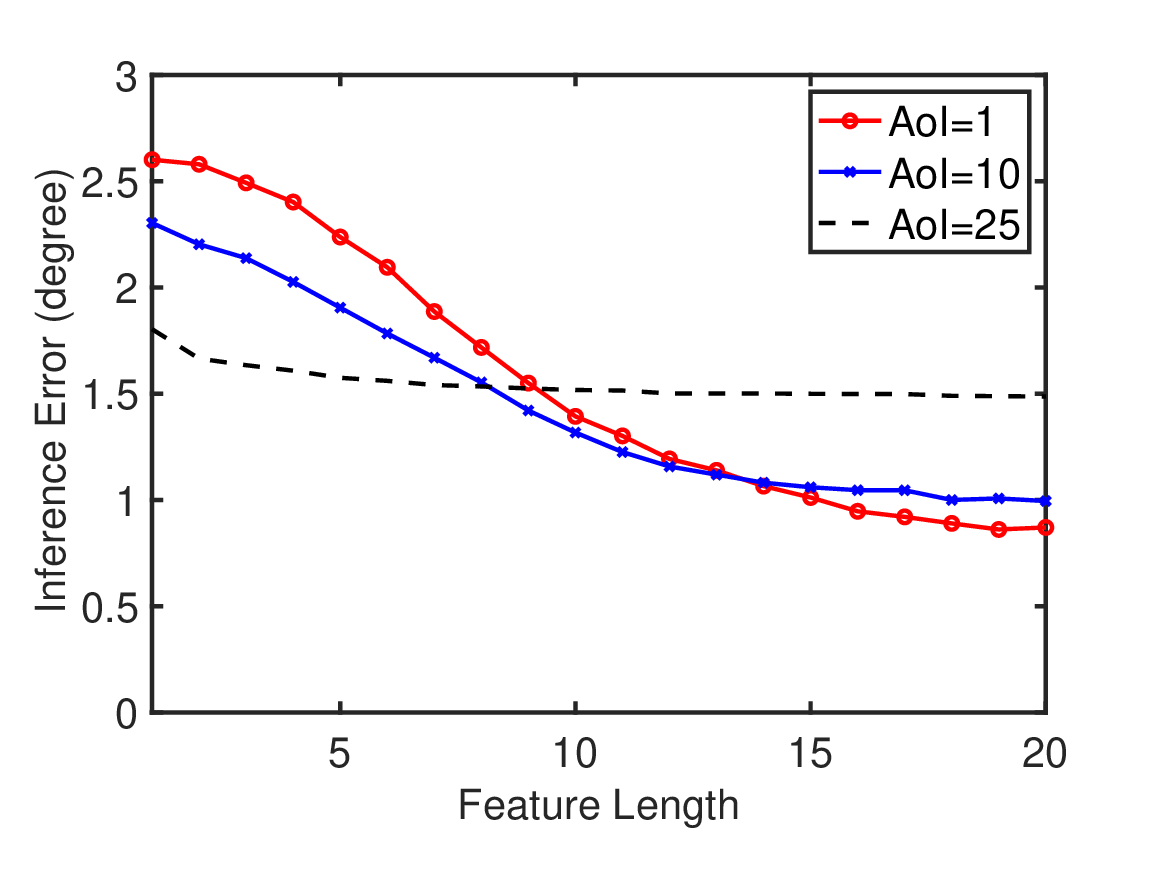}
  \subcaption{\small Inference error vs. Feature length}
\end{subfigure}
   \hspace{10mm} 
\begin{subfigure}[b]{0.6\columnwidth}
\includegraphics[width=\linewidth]
{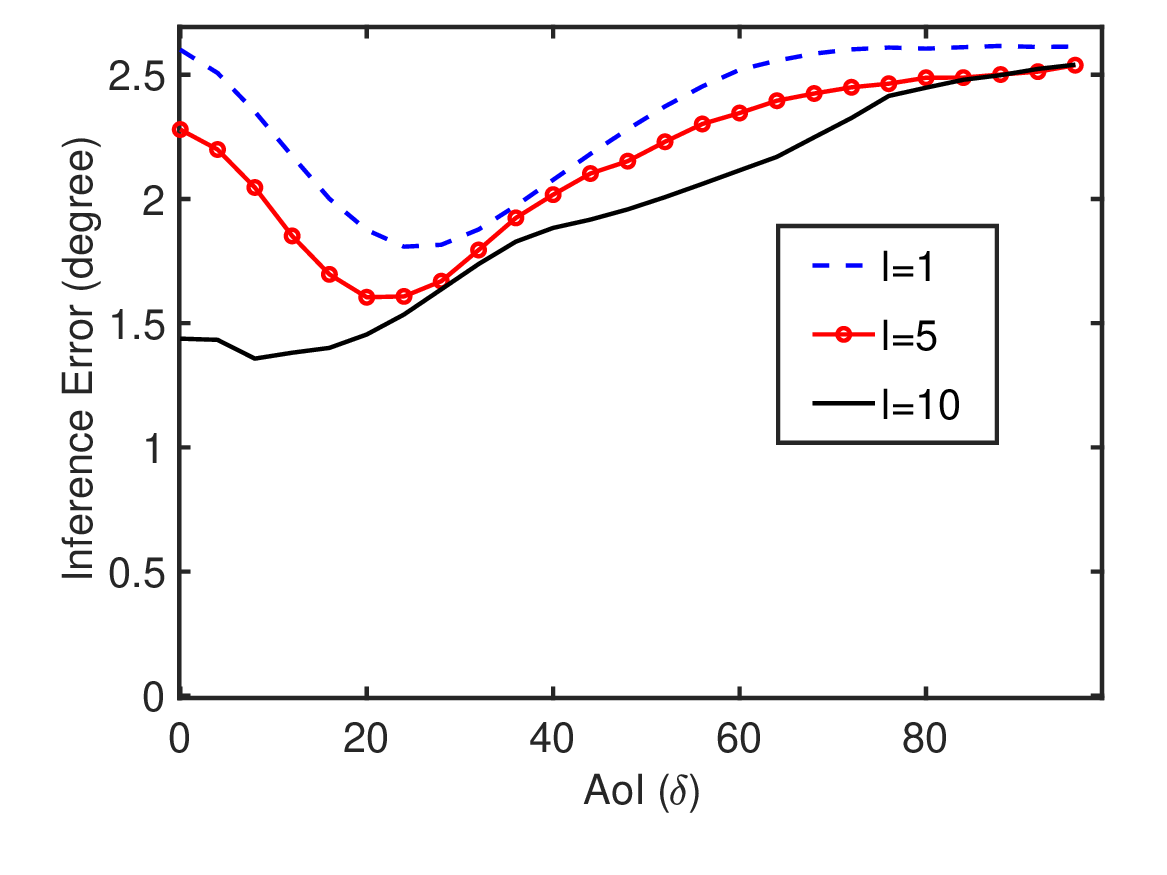}
  \subcaption{\small Inference error vs. AoI}
\end{subfigure}
%
%\begin{subfigure}{0.3\textwidth}
%\includegraphics[width=\textwidth]{Cartpole_Angle.eps}
%  \subcaption{CartPole Angle Prediction}
%\end{subfigure}
%%
\caption{Performance of actuator state prediction in the OpenAI CartPole-v1 task under mechanical response delay: (a) Inference error Vs. Feature length and (b) Inference error Vs. AoI. \ignore{In the OpenAI CartPole-v1 task \cite{brockman2016openai}, the pole angle is predicted by using the cart velocity. Because of the mechanical response delay between cart velocity and pole angle, the inference error are non-monotonic in the AoI.}}
\label{fig:Trainingcartpoleangle}
\end{figure*}

Let $\pi=((S_1, b_1, l_1), (S_2, b_2, l_2), \ldots)$ represent a scheduling policy. We focus on the class of \emph{signal-agnostic} scheduling policies in which each decision is determined without using the knowledge of the signal value of the observed process. A scheduling policy $\pi$ is said to be signal-agnostic, if the policy is independent of $\{(Y_t, X^l_t), t =0,1,2,\ldots\}$. Let $\Pi$ denote the set of all the causal scheduling policies that satisfy the following conditions: (i) the scheduling time $S_i$, the feature position $b_i$, and the feature length $l_i$ are decided based on the current and the historical information available at the scheduler such that $1 \leq l_i \leq B$ and $0 \leq b_i \leq B-l_i$, (ii) the scheduler has access to the inference error function $\mathrm{err}_{\mathrm{inference}}(\cdot)$ and the distribution of $T_i(l)$ for each $l=1, 2, \ldots, B$, and (iii) the scheduler does not have access to the realization of the process $\{(Y_t, X^l_t), t =0,1,2,\ldots\}$. We use $\Pi_{\mathrm{inv}}\subset \Pi$ to denote the set of causal scheduling policies with time-invariant feature length, defined as
\begin{align}\label{PiInv}
 \Pi_{\mathrm{inv}}:=\bigcup_{l=1}^{B} \Pi_l,
\end{align}
where $\Pi_l:= \{\pi \in \Pi: l_1=l_2=\dots=l\}$.

\ignore{\section{System Model and Problem Formulation}\label{Scheduling:SingleUser}

\subsection{Real-time Inference Model}\label{SystemModel1}
\begin{figure}[t]
\centering
\includegraphics[width=0.70\textwidth]{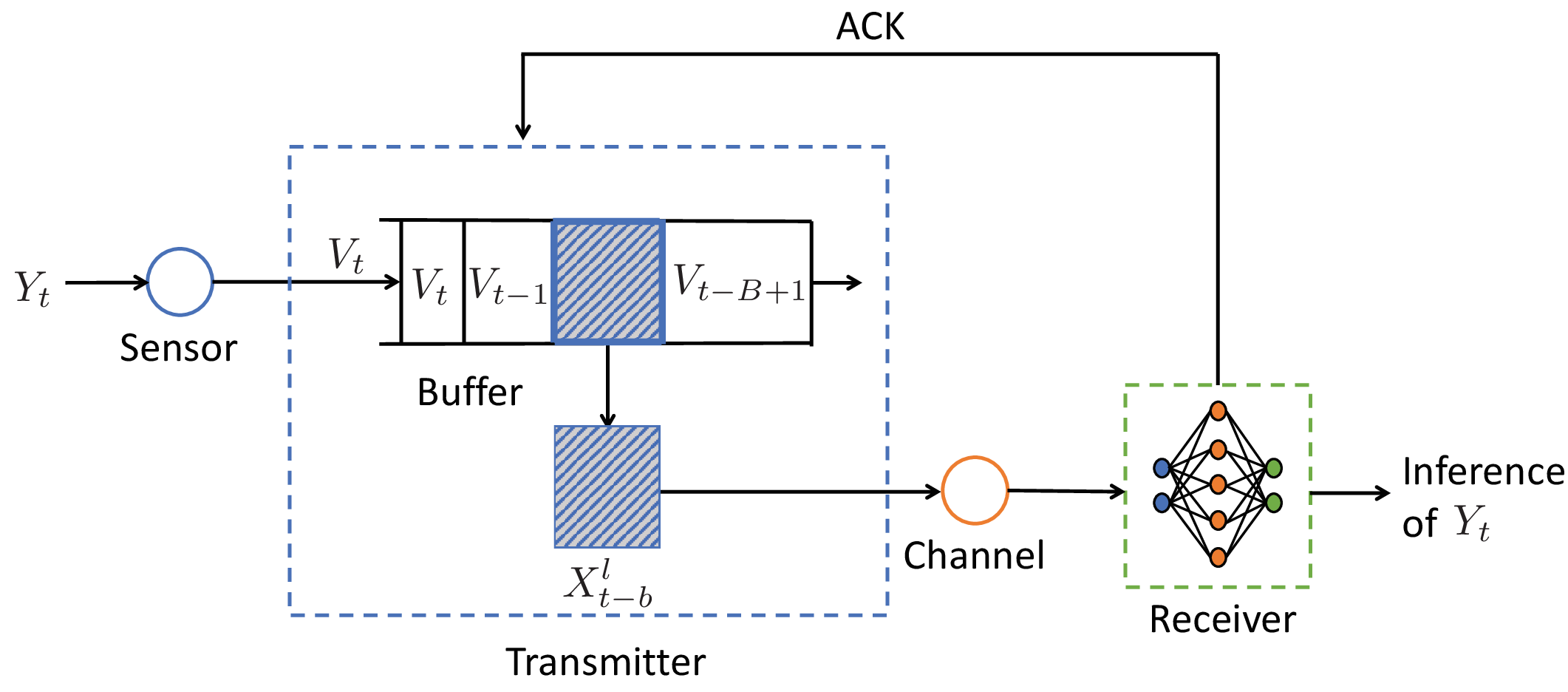}
\caption{\small  A remote inference system, where $X_{t-b}^l:= (V_{t-b}, V_{t-b-1},\ldots, V_{t-b-l+1})$. \label{fig:scheduling}
}
\vspace{-3mm}
\end{figure}
We consider a remote inference system that is composed of a sensor, a transmitter, and a receiver, as illustrated in Fig. \ref{fig:scheduling}. The sensor observes a time-varying target $Y_t \in \mathcal Y$ and feeds its measurements $V_t \in \mathcal V$ to the transmitter. The transmitter progressively generates features $X_{t-b}^l$ from the sensory data $V_t$, where each feature $X_{t-b}^l:= (V_{t-b}, V_{t-b-1},\ldots, V_{t-b-l+1})\in \mathcal V^l$\footnote{$\mathcal V^l$ is the set of all $l$-tuples that take values from $\mathcal V$.} is a temporal sequence taken from $V_t$. For ease of presentation, the temporal sequence length $l$ of feature $X_{t-b}^l$ is called \emph{feature length}. The features are sent one-by-one through a communication channel to the receiver. {\blue Due to communication delays, channel errors, and processing time delays, the features are not received instantly. The most recently received feature is denoted as $X^l_{t-\Delta(t)} = (V_{t-\Delta(t)}, V_{t-\Delta(t)-1}, \ldots, V_{t-\Delta(t)-l+1})$, where the \emph{age of information (AoI)} $\Delta(t)$ represents the difference between the timestamps of the target $Y_t$ and the freshest observation $V_{t-\Delta(t)}$ in feature $X^l_{t-\Delta(t)}$. At the receiver, a neural network associated with feature length $l$ takes the AoI $\Delta(t) \in \mathbb Z^{+}$ and the feature $X^l_{t-\Delta(t)} \in \mathcal V^l$ as inputs and generates an output $a = \phi_{l}(\Delta(t), X^l_{t-\Delta(t)}) \in \mathcal A$, where the neural network is represented by the function $\phi_l: \mathbb Z^{+} \times \mathcal V^l \mapsto \mathcal A$. The performance of the neural network is measured by a loss function $L: \mathcal Y \times \mathcal A \mapsto \mathbb R$, where $L(y, a)$ indicates the incurred loss if the output $a \in \mathcal A$ is used for inference when $Y_t = y$.} The loss function $L$ is determined by the purpose of the application. For example, in softmax regression (i.e., neural network based maximum likelihood classification), the output $a=Q_Y$ is a distribution of $Y_t$ and the loss function $L_{\text{log}}(y, Q_Y)=-\text{log}~Q_Y(y)$ is the negative log-likelihood of the value $Y_t=y$. In neural network based mean-squared estimation, a quadratic loss function $L_2(\mathbf y, \mathbf{\hat y})=\|\mathbf y-\mathbf{\hat y}\|^2_2$ is used, where the action $a=\mathbf{\hat y}$ is an estimate of the target value $Y_t=\mathbf y$ and $\| \mathbf y \|_2$ is the $L^2$ or euclidean norm of the vector $\mathbf y$. 

{\blue Given AoI $\Delta(t)=\delta$, the expected inference error is expressed as a function of the AoI $\delta$ and feature length $l$, which is given by
\begin{align}\label{instantaneous_err1} 
\mathrm{err}_{\mathrm{inference}}(\delta, l):=\mathbb E_{Y, X^l \sim P_{Y_t, X^{l}_{t-\delta}}}\left[L\left(Y,\phi_{l}(X^l, \delta)\right)\right],
\end{align}
where $P_{Y_t, X^{l}_{t-\delta}}$ is the joint distribution of the label $Y_t$ and feature $X^{l}_{t-\delta}$. The inference error $\mathrm{err}_{\mathrm{inference}}(\delta, l)$ can be computed through machine learning experiments. Five experiments are detailed in \cite{ShisherMobihoc}, including (a) video prediction, (b) robotic state prediction, (c) actuator state prediction, (d) channel state information prediction, and (e) temperature prediction. The experimental results in \cite{ShisherMobihoc} suggest that the inference error $\mathrm{err}_{\mathrm{inference}}(\delta, l)$ is a non-increasing function of the feature length $l$ for a fixed AoI $\delta$. This implies that a longer feature leads to a lower inference error. However, the inference error $\mathrm{err}_{\mathrm{inference}}(\delta, l)$ is not necessarily a monotonic function of the AoI $\delta$ for a fixed feature length $l$. In Appendix \ref{information-theoreticAnalysis}, we provide an information-theoretic analysis demonstrating the joint impact of AoI $\delta$ and feature length $l$ on the inference error $\mathrm{err}_{\mathrm{inference}}(\delta, l)$.}

\subsection{Communication model}\label{communication}
The system is time-slotted and starts to operate at time slot $t = 0$. At the beginning of each time slot $t$, the transmitter appends sensory output $V_t$ to a buffer that stores $B$ most recent sensory outputs $(V_t, V_{t-1}, \ldots, V_{t-B+1});$ meanwhile, the oldest output $V_{t-B}$ is removed from the buffer. We assume that the buffer is full initially, such that $B$ signal values $(V_0, V_1, \ldots, V_{-B+1})$ are stored in the buffer at time $t=0$. This ensures that the buffer is always full at any time $t$.\footnote{This assumption does not introduce any loss of generality. If the buffer is no full at time $t=0$, it would not affect our results.} 

In each time slot, the channel is either in an \emph{idle} or \emph{busy} state. If the channel is idle in time slot $t$, the transmitter can select a temporal sequence of sensory data $X_{t-b}^l = (V_{t-b},\ldots,V_{t-b-l+1})$ from the buffer such that $1\leq l\leq B$ and $0\leq b\leq B-l$, and submit the sequence $X_{t-b}^l$ to the channel. In this work, the sequence $X_{t-b}^l$ is called a \emph{feature}. Recall that the temporal sequence length $l$ of feature $X_{t-b}^l$ is called \emph{feature length}. Also, for convenience, the starting position $b$ of feature $X_{t-b}^l$ in the buffer is called \emph{feature position}. This is known as the ``selection-from-buffer" model, which was proposed recently in \cite{ShisherMobihoc}. In contrast to the ``generate-at-will" model \cite{sun2017update, yates2015lazy}, where the transmitter can only select the most recent sensory output $V_t$, the ``selection-from-buffer" model offers greater flexibility by allowing the transmitter to choose multiple sensory outputs (which can be stale or fresh) that are relevant to the target $Y_t$. The number of the signal values selected is determined by the feature length. Feature length represents a trade-off between learning and communications: A longer feature can provide better learning performance, whereas it requires more channel resources (e.g., more time slots or more frequency bandwidth) for sending the feature. This motivated us to study a learning-communication co-design problem that jointly optimizes the feature length, feature position, and transmission scheduling.

The feature length and feature position may vary across the features sent over time. Feature transmissions over the channel are non-preemptive: the channel must finish sending the current feature, before becoming available to transmit the next feature. Suppose that the $i$-th feature $X_{S_i-b_i}^{l_i} = (V_{S_i-b_i}, V_{S_i-b_i-1},\ldots, V_{S_i-b_i-l_i+1})$ is submitted to the channel at time slot $t=S_i$, where $l_i$ is its feature length and $b_i$ is its feature position such that $1\leq l_i\leq B$ and $0\leq b_i\leq B-l_i$. It takes $T_i(l_i)\geq 1$ time slots to send the $i$-th feature over the channel. The $i$-th feature is delivered to the receiver at time slot $D_i=S_i + T_i(l_i)$, where $S_i < D_i \leq S_{i+1}$. The feature transmission time $T_i(l_i)$ depends on the feature length $l_i$. Due to time-varying channel conditions, we assume that, given feature length $l_i=l$, the $T_i(l)$'s are \emph{i.i.d.} random variables, with a finite mean $1\leq \mathbb E[T_i(l)] < \infty$. Once a feature is delivered, an acknowledgment (ACK) is sent back to the transmitter, notifying that the channel has become idle.

In time slot $t$, the $i(t)$-th feature $X_{S_{i(t)}-b_{i(t)}}^{l_{i(t)}}$ is the most recently received feature, where $i(t)=\max_i\{D_i\leq t\}$. The receiver feeds the feature $X_{S_{i(t)}-b_{i(t)}}^{l_{i(t)}}$ to the neural network to infer $Y_t$. \emph{Age of information (AoI)} $\Delta(t)$ is defined as the time difference between the time-stamp of the freshest sensory output $V_{S_{i(t)}-b_{i(t)}}$ in feature $X_{S_{i(t)}-b_{i(t)}}^{l_{i(t)}}$ and the current time $t$, i.e.,  
\begin{align}\label{Def_AoI}
\Delta(t):=t-\max_i \{S_i-b_i: D_i \leq t \}.
\end{align}
Because $D_i < D_{i+1}$, it holds that  
\begin{align}\label{Def_AoI1}
\Delta(t)=t-S_i+b_i,~\text{if}~ D_i \leq t < D_{i+1}.
\end{align}
The initial state of the system is assumed to be $S_0=0, l_0=1, b_0=0, D_0=T_0(l_0)$, and $\Delta(0)$ is a finite constant. Since AoI $\Delta(t)$ and feature length $l_{i(t)}$ can change over time, multiple neural networks need to be utilized to cover all possible values of feature length.

\subsection{Learning and communication co-design problem}
Let $\pi=((S_1, b_1, l_1), (S_2, b_2, l_2), \ldots)$ represent a scheduling policy and $\Pi$ denote the set of all the causal scheduling policies that satisfy the following conditions: (i) the scheduling time $S_i$, the feature position $b_i$, and the feature length $l_i$ are decided based on the current and the historical information available at the transmitter such that $1 \leq l_i \leq B$ and $0 \leq b_i \leq B-l_i$ and (ii) the scheduler has access to the inference error function $\mathrm{err}_{\mathrm{inference}}(\cdot)$ and the distribution of $T_i(l)$ for each $l=1, 2, \ldots, B$. We use $\Pi_{\mathrm{inv}}$ to denote the set of causal scheduling policies with time-invariant feature length, defined as
\begin{align}
 \Pi_{\mathrm{inv}}:=\bigcup_{l=1}^{B} \Pi_l,
\end{align}
where $\Pi_l$ is given by
\begin{align}
\Pi_l:= \{\pi \in \Pi: l_1=l_2=\dots=l\}.
\end{align}
Under policy $\pi$, the time-averaged expected inference error $\bar p_{\pi}$ is given by
\begin{align}
\bar p_{\pi}=\limsup_{T\rightarrow \infty}\frac{1}{T} \mathbb{E}_{\pi} \left[ \sum_{t=0}^{T-1} \mathrm{err}_{\mathrm{inference}}(\Delta(t), l_{i(t)})\right],
\end{align}
where $\mathrm{err}_{\mathrm{inference}}(\Delta(t), l_{i(t)})$ is the expected inference error at time slot $t$.

We consider two problems: (i) The goal of the first problem is to find an optimal policy from $\Pi_{\mathrm{inv}}$ that minimizes the time-averaged inference error and (ii) the second problem is to find an optimal policy from the set of all causal scheduling policies $\Pi$ that minimizes the time-averaged inference error.}

{\section{Preliminaries: Impacts of Feature Length and AoI on Inference Error}\label{information-theoreticAnalysis}
In this section, we adopt an information-theoretic approach that was developed recently in \cite{ShisherMobihoc} to show the impact of feature length $l$ and AoI $\delta$ on the inference error $\mathrm{err}_{\mathrm{inference}}(\delta, l)$.

\subsection{Information-theoretic Metrics for Training and Inference Errors}
Training error $\mathrm{err}_{\mathrm{training}}(\delta, l)$ is expressed as a function of $\delta$ and $l$, given by
\begin{align}\label{trainingError}
\mathrm{err}_{\mathrm{training}}(\delta, l)=\mathbb{E}_{Y,X^l\sim P_{\tilde Y_0, \tilde X^l_{-\delta}}}[L(Y,\phi_{l}(\delta, X^l))], 
\end{align}
where $\phi_l$ a trained neural network used in \eqref{instantaneous_err1} and $P_{\tilde Y_0, \tilde X^l_{-\delta}}$ is the joint distribution of the target $\tilde Y_0$ and the feature $\tilde X^l_{-\delta}$ in the training dataset. The training error $\mathrm{err}_{\mathrm{training}}(\delta, l)$ is lower bounded by 
\begin{align}\label{eq_TrainingErrorLB}
H_L(\tilde Y_0| \tilde X^l_{-\delta})=\min_{\phi_{l} \in \Phi} \mathbb{E}_{Y,X^l\sim P_{\tilde Y_0, \tilde X^l_{-\delta}}}[L(Y,\phi_{l}(\delta, X^l))],
\end{align} 
where $\Phi=\{\phi_{l} : \mathbb Z^{+} \times \mathcal V^l \mapsto \mathcal A\}$ is the set of all functions that map from $\mathbb Z^{+} \times \mathcal V^l$ to $\mathcal A$. Because the trained neural network $\phi_l$ in \eqref{trainingError} satisfies $\phi_l \in \Phi$, $H_L(\tilde Y_0| \tilde X^l_{-\delta}) \leq \mathrm{err}_{\mathrm{training}}(\delta, l)$. The lower bound in \eqref{eq_TrainingErrorLB} has an information-theoretical interpretation \cite{ShisherMobihoc, Dawid2004, farnia2016minimax, Dawid1998}: It is a generalized conditional entropy of a random variable $\tilde Y_0$ given $\tilde X^l_{-\delta}$ associated to the loss function $L$.  For notational simplicity, we call $H_L(Y|X)$ an $L$-conditional entropy of a random variable $Y$ given $X$. The \emph{$L$-entropy} of a random variable $Y$ is defined as \cite{Dawid2004, farnia2016minimax}
\begin{align}\label{eq_Lentropy}
H_L(Y) = \min_{a\in\mathcal A} \mathbb{E}_{Y \sim P_{Y}}[L(Y,a)].
\end{align} 
The optimal solutions to \eqref{eq_Lentropy} may not be unique.  Let $a_{P_Y}$ denote an optimal solution to \eqref{eq_Lentropy}, which is called a \emph{Bayes action} \cite{Dawid2004}.
Similarly, the $L$-conditional entropy of $Y$ given $X=x$ is defined as \cite{shisher2021age,ShisherMobihoc,Dawid2004,farnia2016minimax}
\begin{align}\label{given_L_condentropy}
H_L(Y| X=x)= \min_{a \in\mathcal A}\! \mathbb E_{Y\sim P_{Y| X=x}} [L(Y, a)]
\end{align}
and the $L$-conditional entropy of $Y$ given $X$ is given by \cite{shisher2021age,ShisherMobihoc,Dawid2004, farnia2016minimax}
\begin{align}\label{eq_cond_entropy1}
H_L(Y|X)=\sum_{x \in \mathcal X} P_X(x) H_L(Y| X=x).
\end{align}

The inference error $\mathrm{err}_{\mathrm{inference}}(\delta, l)$ can be approximated as the following $L$-conditional cross entropy 
\begin{align}\label{inferenceConditionCross}
\!\!\!\!\!\!&H_L(P_{Y_t|X^l_{t-\delta}}; P_{\tilde Y_0|\tilde X^l_{-\delta}}|P_{X^l_{t-\delta}})\nonumber\\
\!\!\!\!\!\!=& \sum_{x \in \mathcal X^l} P_{X^l_{t-\delta}}(x) \mathbb E_{Y \sim P_{Y_t|X^l_{t-\delta}=x}}\left[L\left(Y, a_{P_{\tilde Y_0|\tilde X^l_{-\delta}=x}}\right)\right],\!\!\!
\end{align}
where the $L$-conditional cross entropy $H_L(P_{Y|X}; P_{\tilde Y|\tilde X} | P_X)$ is defined as \cite{ShisherMobihoc}
\begin{align}
&H_L(P_{Y|X}; P_{\tilde Y| \tilde X} | P_X)\nonumber\\
=&\sum_{x \in \mathcal X} P_X(x) \mathbb E_{Y \sim P_{Y|X=x}}\left[L\left(Y, a_{P_{\tilde Y|\tilde X=x}}\right)\right].
\end{align}
If training algorithm considers sets of large and wide neural networks such that $a_{P_{\tilde Y_0|\tilde X^l_{-\delta}=x}}$ and $\phi_l(\delta, x)$ for all $\delta \in \mathbb Z^{+}$ and $x\in \mathcal X^l$ are close to each other, then the difference between the inference error $\mathrm{err}_{\mathrm{inference}}(\delta, l)$ and the $L$-conditional cross entropy $H_L(P_{Y_t|X^l_{t-\delta}}; P_{\tilde Y_0|\tilde X^l_{-\delta}}|P_{X^l_{t-\delta}})$ is small \cite{ShisherMobihoc}. Compared to $\mathrm{err}_{\mathrm{inference}}(\delta, l)$, the $L$-conditional cross entropy $H_L(P_{Y_t|X^l_{t-\delta}}; P_{\tilde Y_0|\tilde X^l_{-\delta}}|P_{X^l_{t-\delta}})$ are mathematically more convenient to analyze, as we will see next. 

\subsection{Information-theoretic Monotonicity Analysis}
The following lemma interprets the monotonicity of the $L$-conditional entropy $H_L(\tilde Y_0| \tilde X^l_{-\delta})$ and the $L$-conditional cross entropy $H_L(P_{Y_t|X^l_{t-\delta}}; P_{\tilde Y_0|\tilde X^l_{-\delta}}|P_{X^l_{t-\delta}})$ with respect to the feature length $l$.

\begin{lemma}\label{lemma2}
The following assertions are true:
\begin{itemize}
\item[(a)] Given $\delta \geq 0$, $H_L(\tilde Y_0| \tilde X^l_{-\delta})$ is a non-increasing function of $l$, i.e., for all $1 \leq l_1 \leq l_2$
\begin{align}\label{lemma2e1}
H_L(\tilde Y_0| \tilde X^{l_2}_{-\delta}) \leq H_L(\tilde Y_0| \tilde X^{l_1}_{-\delta}).
\end{align}
\item[(b)] Given $\beta \geq 0$, if for all $l=1,2, \ldots,$ and $x \in \mathcal V^l$
    \begin{align}\label{condition2}
     &\sum_{x \in \mathcal X^l} P_{X^l_{t-\delta}}(x)\sum_{y \in \mathcal Y} (P_{Y_{t}|X^l_{t-\delta}=x}(y)-P_{\tilde Y_0| \tilde X^l_{-\delta}=x}(y))^2 \nonumber\\
     &\leq  \beta^2,
    \end{align}
then for all $1 \leq l_1 \leq l_2$
    \begin{align}\label{Eq_lemma}
       &H_L(P_{Y_t|X^{l_2}_{t-\delta}}; P_{\tilde Y_0|\tilde X^{l_2}_{-\delta}}|P_{X^{l_2}_{t-\delta}}) \nonumber\\
       \leq& H_L(P_{Y_t|X^{l_1}_{t-\delta}}; P_{\tilde Y_0|\tilde X^{l_1}_{-\delta}}|P_{X^{l_1}_{t-\delta}})+O(\beta).
    \end{align}
\end{itemize}
\end{lemma}

\begin{proof}
Lemma 1 can be proven by using the data processing inequality for $L$-conditional entropy \cite[Lemma 12.1]{Dawid1998} and a local information geometric analysis. \ifreport See Appendix B of the supplementary material for the details. \else Detailed proof is provided in our Technical report \cite{technical_report}. \fi
\end{proof}

Lemma \ref{lemma2}(a) demonstrates that for a given AoI value $\delta$, the $L$-conditional entropy $H_L(\tilde Y_0| \tilde X^l_{-\delta})$ decreases as the feature length $l$ increases. This is due to the fact that a longer feature provides more information, consequently leading to a lower $L$-conditional entropy. Additionally, as indicated in Lemma \ref{lemma2}(b), when the conditional distributions in training and inference data are close to each other (i.e., when $\beta$ in \eqref{condition2} is close to $0$), the $L$-conditional cross entropy $H_L(P_{Y_t|X^{l}_{t-\delta}}; P_{\tilde Y_0|\tilde X^{l}_{-\delta}}|P_{X^{l}_{t-\delta}})$ is close to a non-increasing function of the feature length $l$. This information-theoretic analysis clarifies the experimental results depicted in Fig. \ref{fig:Trainingcsi}(a) and Fig. \ref{fig:Trainingcartpoleangle}(a), where the inference error diminishes with the increasing feature length. 

The monotonicity of the $L$-conditional cross entropy $H_L(P_{Y_t|X^l_{t-\delta}}; P_{\tilde Y_0|\tilde X^l_{-\delta}}|P_{X^l_{t-\delta}})$ with respect to the AoI $\delta$ are explained in Theorem 3 of \cite{ShisherMobihoc} and in \cite{Shisher2023Timely}. This result is restated in Lemma 2 below for the sake of completeness. 

\begin{definition}[\textbf{$\epsilon$-Markov Chain} \cite{ShisherMobihoc, Shisher2023Timely}]
Given $\epsilon \geq 0$, a sequence of three random variables $Y, X,$ and $Z$ is said to be an \emph{$\epsilon$-Markov chain}, denoted as $Y \overset{\epsilon} \leftrightarrow X \overset{\epsilon} \leftrightarrow Z$, if
\begin{align}\label{epsilon-Markov-def}
I_{\mathrm{log}}(Y;Z|X)=\mathbb E_{X, Z \sim P_{X, Z}} \left [ D_{\mathrm{log}}\left(P_{Y|X,Z} || P_{Y|X} \right)\right] \leq \epsilon^2,
\end{align}
where
\begin{align}\label{chi-divergence-def}
D_{\mathrm{log}}(P_Y ||Q_Y)=\sum_{y \in \mathcal{Y}} P_Y(y) \mathrm{log} \frac{P_Y(y)}{Q_Y(y)}
\end{align}
is KL-divergence and $I_{\mathrm{log}}(Y;Z|X)$ is Shannon conditional mutual information.
\end{definition}

\begin{lemma}\label{lemma1} \cite{ShisherMobihoc, Shisher2023Timely}
If $Y_t \overset{\epsilon}\leftrightarrow X^l_{t-\mu} \overset{\epsilon}\leftrightarrow X^l_{t-\mu-\nu}$ is an $\epsilon$-Markov chain for all $\mu, \nu \geq 0$ and \eqref{condition2} holds, then for all $0 \leq \delta_1\leq \delta_2$
\begin{align}\label{eq_theorem3}
 &H_L(P_{Y_t|X^l_{t-\delta_1}}; P_{\tilde Y_0|\tilde X^l_{-\delta_1}}|P_{X^l_{t-\delta_1}})\nonumber\\
 \leq& H_L(P_{Y_t|X^l_{t-\delta_2}}; P_{\tilde Y_0|\tilde X^l_{-\delta_2}}|P_{X^l_{t-\delta_2}})+O\big(\max\{\epsilon, \beta\}\big).
\end{align}
\end{lemma}\

Lemma \ref{lemma1} implies that the monotonic behavior of $H_L(P_{Y_t|X^l_{t-\delta}}; P_{\tilde Y_0|\tilde X^l_{-\delta}}|P_{X^l_{t-\delta}})$ with respect to AoI $\delta$ is characterized by two key parameters: $\epsilon$ in the $\epsilon$-Markov chain model and the parameter $\beta$. When $\epsilon$ is small, the sequence of target and feature random variables approximates a Markov chain. Consequently, $H_L(P_{Y_t|X^l_{t-\delta}}; P_{\tilde Y_0|\tilde X^l_{-\delta}}|P_{X^l_{t-\delta}})$ becomes non-decreasing with respect to AoI $\delta$ provided that $\beta$ is close to $0$. Conversely, if $\epsilon$ is significantly large, then $H_L(P_{Y_t|X^l_{t-\delta}}; P_{\tilde Y_0|\tilde X^l_{-\delta}}|P_{X^l_{t-\delta}})$ can be far from a monotonic function of $\delta$. This findings provide an explanation for the patterns observed in the experimental results shown in Figs. \ref{fig:Trainingcsi}(b) to \ref{fig:Trainingcartpoleangle}(b). Shannon's interpretation of Markov sources in his seminal work \cite{Shannon1948} indicates that as the sequence length $l$ grows larger, the tuple $(Y_t, X^l_{t-\mu}, X^l_{t-\mu-\nu})$ tends to resemble a Markov chain more closely. Hence, according to Lemma \ref{lemma1}, the inference error approaches to a non-decreasing function of AoI $\delta$ as feature length $l$ increases. As illustrated in Figs. \ref{fig:Trainingcsi}(b)-\ref{fig:Trainingcartpoleangle}(b), the inference error converges to a non-decreasing function of AoI $\delta$ as feature length $l$ increases. 
}
\section{Learning and Communications Co-design: Single Source Case}\label{Scheduling:SingleUser}
Let $d(t)$ denote the feature length of the most recently received feature in time slot $t$. The time-averaged expected inference error under policy $\pi=((S_1, b_1, l_1), (S_2, b_2, l_2), \ldots)$ is expressed as
\begin{align}
\bar p_{\pi}=\limsup_{T\rightarrow \infty}\frac{1}{T} \mathbb{E}_{\pi} \left[ \sum_{t=0}^{T-1} \mathrm{err}_{\mathrm{inference}}(\Delta(t), d(t))\right],
\end{align}
where $\bar p_{\pi}$ is denoted as the time-averaged inference error, and $\mathrm{err}_{\mathrm{inference}}(\Delta(t), d(t))$ is the expected inference error at time $t$ corresponding to the system state $(\Delta(t), d(t))$. In this section, we slove two problems. The first one is to find an optimal policy that minimizes the time-averaged expected inference error among all the causal policies in $\Pi_{\mathrm{inv}}$ that consider time-invariant feature length. Another problem is to find an optimal policy that minimizes the time-averaged expected inference error among all the causal policies in $\Pi$.

\subsection{Time-invariant Feature Length}\label{inv}
We first find an optimal policy that minimizes the time-averaged inference error among all causal policies with time-invariant feature length in $\Pi_{\mathrm{inv}}$ defined in \eqref{PiInv}: 
\begin{align}\label{scheduling_problem1}
\bar p_{\mathrm{inv}}\!\!=\!\!\inf_{\pi \in \Pi_{\mathrm{inv}}}  \!\!\limsup_{T\rightarrow \infty}\frac{1}{T} \mathbb{E}_{\pi} \!\!\left[ \sum_{t=0}^{T-1} \!\mathrm{err}_{\mathrm{inference}}(\Delta(t), d(t))\right], 
\end{align}
where $\bar p_{\mathrm{inv}}$ is the optimum value of \eqref{scheduling_problem1}. The problem \eqref{scheduling_problem1} is an infinite time-horizon average-cost semi-Markov decision process (SMDP). Such problems are often challenging to solve analytically or with closed-form solutions. The per-slot cost function $\mathrm{err}_{\mathrm{inference}}(\Delta(t), d(t))$ in \eqref{scheduling_problem1} depends on two variables: the AoI $\Delta(t)$ and the feature length $d(t)$. Prior studies \cite{SunNonlinear2019, SunTIT2020, orneeTON2021, klugel2019aoi, Tripathi2019, kadota2018optimizing, Kadota2018, Kadota2019, sun2017update} have considered linear and non-linear monotonic AoI functions. Due to the fact that (i) the cost function in \eqref{scheduling_problem1} depends on two variables and (ii) is not necessarily monotonic with respect to AoI, finding an optimal solution is challenging and the existing scheduling policies cannot be directly applied to solve \eqref{scheduling_problem1}. Therefore, it is necessary to develop a new scheduling policy that can address the complexities of \eqref{scheduling_problem1}.

Surprisingly, we get a closed-form solution of \eqref{scheduling_problem1}. To present the solution, we define a function $\gamma_l(\delta, d)$ as
\begin{align}\label{gittins}
\gamma_l(\delta, d):= \inf_{\tau \in \{1, 2, \ldots\}} \frac{1}{\tau}\sum_{j=0}^{\tau-1}\mathbb E\bigg[\mathrm{err}_{\mathrm{inference}}\bigg(\delta+j+T_1(l), d\bigg) \bigg].
\end{align}

\begin{theorem}\label{theorem1}
If $T_i(l)$'s are i.i.d. with a finite mean $\mathbb E[T_i(l)]$ for each $l=1,2, \ldots, B$, then there exists an optimal solution $\pi^*=((S^*_1, b^*_1, l^*), (S^*_2, b^*_2, l^*), \ldots) \in \Pi_{\mathrm{inv}}$ to \eqref{scheduling_problem1} that satisfies:
\begin{itemize}
\item[(a)] The optimal feature position in $\pi^*$ is time-invariant, i.e., $b^*_1=b^*_2= \dots=b^*$. The optimal feature length $l^*$ and the optimal feature position $b^*$ in $\pi^*$ are given by
\begin{align}\label{optimal_buffer_length_1}
(l^*, b^*)=\argmin_{\substack{l \in \mathbb Z, b \in \mathbb Z\\1 \leq l \leq B, 0 \leq b \leq B-l}} \beta_{b, l},
\end{align}
where $\beta_{b, l}$ is the unique root of equation
\begin{align}\label{bisection}
&\mathbb{E}\left[\sum_{t=D_i(\beta_{b,l})}^{D_{i+1}(\beta_{b,l})-1}  \mathrm{err}_{\mathrm{inference}}(\Delta_b(t), l)\right]\nonumber\\
&- \beta_{b,l}~ \mathbb{E}\bigg[D_{i+1}(\beta_{b,l})-D_{i}(\beta_{b,l})\bigg]=0, 
\end{align}
$D_{i}(\beta_{b,l})=S_{i}(\beta_{b,l})+T_{i}(l)$, $\Delta_b(t)=t-S_{i}(\beta_{b,l})+b$, the sequence $(S_{1}(\beta_{b,l}), S_{2}(\beta_{b,l}), \ldots)$ is determined by
\begin{align}\label{OptimalPolicy_fixed_length_fixed_buffer}
\!\!\!\!S_{i+1}(\beta_{b,l})\!=\! \min_{t \in \mathbb Z}\!\big\{ t\! \geq\! D_i(\beta_{b,l}): \gamma_l(\Delta_b(t), l) \!\geq\! \beta_{b,l} \big\},
\end{align}
and the function $\gamma_l(\cdot)$ is defined in \eqref{gittins}.

\item[(b)] The optimal scheduling time $S^*_{i+1}$ in $\pi^*$ is determined by 
\begin{align}\label{OptimalWaitingTime2}
\!\!\!\!S^*_{i+1}\!\!=\! \min_{t \in \mathbb Z}\big\{ t \!\geq\! S^*_i\!+\!T_i(l^*)\!\!: \!\gamma_{l^*}(\Delta_{b^*}(t), l^*)\!\! \geq\! \bar p_{\mathrm{inv}}\big\},
\end{align}
where $\Delta_{b^*}(t)=t-S^*_i+b^*$ is the AoI at time $t$. The optimal objective value $\bar p_{\mathrm{inv}}$ of \eqref{scheduling_problem1} is 
\begin{align}\label{optimal_objective1}
\bar p_{\mathrm{inv}}=\min_{\substack{l \in \mathbb Z, b \in \mathbb Z\\1 \leq l \leq B, 0 \leq b \leq B-l}} \beta_{b, l}.
\end{align}
\end{itemize}
\end{theorem}
We prove Theorem \ref{theorem1} in two steps: (i) We find $B$ policies, each of which is optimal among the set of policies $\Pi_l$ where $l=1, 2, \ldots, B$. After that
(ii) we select the policy that results in the minimum average inference error among the $B$ policies. \ifreport See Appendix C of the supplementary material for details. \else Due to space limit, we refer readers to our technical report \cite{technical_report} for detailed proof of Theorem \ref{theorem1}.
\fi

Theorem \ref{theorem1} implies that the optimal scheduling policy has a nice structure. According to Theorem \ref{theorem1}(a), the feature position $b^*_i$ is constant for all $i$-th features, i.e., $b_1^*=b_2^*=\ldots=b^*$. The optimal feature length $l^*$ and the optimal feature position $b^*$ are pre-computed by solving \eqref{optimal_buffer_length_1} and then used in real-time. The parameter $\beta_{b, l}$ in \eqref{optimal_buffer_length_1} is the unique root of \eqref{bisection}, which is solved by using low-complexity algorithms, e.g., bisection search, newtons method, and fixed point iteration \cite{orneeTON2021}. Theorem \ref{theorem1}(b) implies that the optimal scheduling time $S^*_{i+1}$ follows a threshold policy. Specifically, a feature is transmitted in time-slot $t$ if the following two conditions are satisfied: (i) The channel is idle in time-slot $t$ and (ii) the value $\gamma_{l^*}(\Delta(t), l^*)$ exceeds the optimal objective value $\bar p_{\mathrm{inv}}$ of \eqref{scheduling_problem1}. The optimal objective value $\bar p_{\mathrm{inv}}$ is obtained from \eqref{optimal_objective1}. Our threshold-based scheduling policy has a significant distinction from previous threshold-based policies studied in the literature, such as \cite{SunNonlinear2019, SunTIT2020, orneeTON2021, klugel2019aoi}. In these prior works, the threshold function used to determine the scheduling time is based solely on the AoI value and is non-decreasing with respect to AoI. However, in our proposed strategy, (i) the threshold function $\gamma_l(\cdot)$ depends on both the AoI value and the feature length and (ii) the threshold function $\gamma_l(\cdot)$ can be non-monotonic with respect to AoI.

\subsubsection{Monotonic AoI Cost function} 
Consider a special case where the inference error $\mathrm{err}_{\mathrm{inference}}(\delta, l)$ is a non-decreasing function of $\delta$ for every feature length $l$. A simplified solution can be derived for this specific case of \eqref{scheduling_problem1}. In this scenario, the optimal feature position is $b^*=0$, and the threshold function $\gamma_l(\cdot)$ defined in \eqref{gittins} becomes:
\begin{align}
    \gamma_l(\delta, d)=\mathbb E\left[\mathrm{err}_{\mathrm{inference}}\bigg(\delta+T_1(l), d\bigg)\right].
\end{align} 
In this special case of monotonic AoI cost function, \eqref{OptimalWaitingTime2} can be rewritten as a threshold policy of the AoI $\Delta(t)$ in the form of $\Delta(t) \geq w(l^*, \bar p_{\mathrm{inv}})$, where $w(l, \beta)$ is defined as:
\begin{align}
    w(l,\beta)=\inf\bigg\{\delta \geq 0: \mathbb E\left[\mathrm{err}_{\mathrm{inference}}\bigg(\delta+T_1(l), l\bigg)\right] \geq \beta \bigg\}. 
\end{align}
However, when $\mathrm{err}_{\mathrm{inference}}(\delta, l)$ is not monotonic with respect to AoI $\delta$, \eqref{OptimalWaitingTime2} cannot be reformulated as a threshold policy of the AoI $\Delta(t)$. This is a key difference with earlier studies \cite{SunNonlinear2019, klugel2019aoi, Tripathi2019}.

\subsubsection{Connection with Restart-in-state Problem} Consider another special case in which all features take $1$ time-slot for transmission. For this special case, the threshold function $\gamma_l(\cdot)$ defined in \eqref{gittins} becomes  
\begin{align}\label{gittinsSpecial}
    \gamma_l(\delta, d)=\inf_{\tau \in \{1, 2, \ldots\}} \frac{1}{\tau}\sum_{j=0}^{\tau-1}\mathbb E\bigg[\mathrm{err}_{\mathrm{inference}}\bigg(\delta+j+1, d\bigg) \bigg].
\end{align} 
This special case of \eqref{scheduling_problem1} is a restart-in-state problem \cite[Chapter 2.6.4] {gittins2011multi}. This is because whenever a feature with the optimal feature length $l^*$ and from the optimal feature position $b^*$ is transmitted, AoI value restarts from $b^*+1$ in the next time slot. For this restart-in-state problem, the optimal sending time follows a threshold policy \cite[Chapter 2.6.4] {gittins2011multi}. Specifically, a feature is transmitted if 
\begin{align}\label{ineqrestart}
    h(\Delta_{b^*}(t+1), l^*) \geq h(b^*+1, l^*),
\end{align}
where the relative value function $h(\delta, l^*)$ of the restart-in-state problem is given by
\begin{align}\label{Bellmanrestart}
h(\delta, l^*)=&\min_{Z \in \{0, 1, \ldots\}}\mathbb E\left[\sum_{k=0}^{Z}\!\! \bigg(\mathrm{err}_{\mathrm{inference}}(\delta+k, l^*)- \bar p_{\mathrm{inv}}\bigg)\right]\nonumber\\
&\quad\quad\quad\quad+h(b^*+1,l^*).
\end{align}
By using \eqref{Bellmanrestart}, we can show that \eqref{ineqrestart} is equivalent to 
\begin{align}
    \gamma_{l^*}(\Delta_{b^*}(t), l^*) \geq \bar p_{\mathrm{inv}}.
\end{align}
where the function $\gamma_l(\delta, d)$ is defined in \eqref{gittinsSpecial}. This connection between the restart-in-state problem and AoI minimization was unknown before. The original problem considers more general $T_i(l)$, which can be considered as a restart-in-random state problem. This is because whenever $i$-th feature with optimal feature length $l^*$ and from optimal feature position $b^*$ is transmitted, AoI restarts from a random value $b^*+T_i(l^*)$ after $T_i(l^*)$ time slots.

\ignore{Re-arranging \eqref{gittins_subtract} and by definition, $\gamma_l(\delta, d)$ is given by
\begin{align}\label{gittins1}
&\gamma_l(\delta, l_d) \nonumber\\
=&\bigg\{r:\!\!\!\!\sup_{\nu \in \mathfrak M, \nu \neq 0}\!\!\!\! \mathbb E\left[ \sum_{k=0}^{\nu-1} [r-\mathrm{err}_{\mathrm{inference}}(\Delta(t+k), d)]\bigg| \Delta(t)\!=\!\delta \right]\!\!=\!0\bigg\}.
\end{align}}
%It can be shown that $\gamma_l(\delta, d)$ defined in \eqref{gittins} is the solution to \eqref{gittins_subtract}. Hence, we call $\gamma_l(\delta, d)$ a Gittins index.

\subsection{Time-variant Feature Length}\label{vary}

Now, we find an optimal scheduling policy that minimizes time-averaged inference error among all causal policies in $\Pi$: 
\begin{align}\label{scheduling_problem}
\bar p_{opt}\!\!=\!\!\inf_{\pi \in \Pi}  \limsup_{T\rightarrow \infty}\frac{1}{T} \mathbb{E}_{\pi} \left[ \sum_{t=0}^{T-1} \mathrm{err}_{\mathrm{inference}}(\Delta(t), d(t))\right], 
\end{align}
where $\mathrm{err}_{\mathrm{inference}}(\Delta(t), d(t))$ is the inference error at time slot $t$ and $\bar p_{opt}$ is the optimum value of \eqref{scheduling_problem}. Because $\Pi_{\mathrm{inv}}\subset \Pi$, 
\begin{align}
\bar p_{opt} \leq \bar p_{\mathrm{inv}},
\end{align}
where $\bar p_{\mathrm{inv}}$ is the optimum value of \eqref{scheduling_problem1}. Like \eqref{scheduling_problem1}, problem \eqref{scheduling_problem} can also be expressed as an infinite time-horizon average-cost SMDP. \ignore{The basic components (decision times, states, actions, and associated transition probabilities) of the SMDP are described in Appendix \ref{ptheorem2}.} Note that \eqref{scheduling_problem} is more complex SMDP than \eqref{scheduling_problem1} because the feature length in \eqref{scheduling_problem} is allowed to vary over time.  

The optimal policy can be determined by using a dynamic programming method associated with the average cost SMDP \cite{puterman2014markov, bertsekasdynamic}. There exists a function $h(\cdot)$ such that for all $\delta \in \mathbb Z^{+}$ and $0 \leq d \leq B$, the optimal objective value $\bar p_{opt}$ of \eqref{scheduling_problem} satisfies the following Bellman equation:
\begin{align}\label{optimal_objective2notsimple}
&h(\delta, d)\nonumber\\
=&\min_{\substack{Z \in \{0, 1, \ldots\} \\l \in \mathbb Z:1\leq l \leq B \\ b \in \mathbb Z:0 \leq b \leq B-l}}\mathbb E\left[\sum_{k=0}^{Z+T_{1}(l)-1}\!\! \bigg(\mathrm{err}_{\mathrm{inference}}(\delta+k, d)- \bar p_{opt}\bigg)\right]\nonumber\\
&\quad\quad\quad\quad\quad\quad+\mathbb E[h(T_{1}(l)+b, l)].
\end{align}
Let $(Z^*(\delta, d), l^*(\delta, d), b^*(\delta, d))$ be the optimal solution to the Bellman equation \eqref{optimal_objective2notsimple}. There exists an optimal solution $\pi^*=((S^*_{1}, b^*_{1}, l^*_1), (S^*_{2}, b^*_{2}, l^*_2), \ldots) \in \Pi$ to \eqref{scheduling_problem}, determined by 
\begin{align}
   l^*_{i+1}&=l^*(T_i(l^*_i)+b^*_i, l^*_i), \\
   b^*_{i+1}&=b^*(T_i(l^*_i)+b^*_i, l^*_i), \\
   S^*_{i+1}&=S^*_i+T_i(l^*_i)+Z^*(T_i(l^*_i)+b^*_i, l^*_i),
\end{align}
where $Z^*(T_i(l^*_i)+b^*_i, l^*_i)$ is the optimal waiting time for sending the $(i+1)$-th feature after the $i$-th feature is delivered. 

To get the optimal policy $\pi^*$, we need to solve \eqref{optimal_objective2notsimple}. Solving \eqref{optimal_objective2notsimple} is complex as it requires joint optimization of three variables. Moreover, an optimal solution obtained by the dynamic programming method provides no insight. We are able to simplify \eqref{optimal_objective2notsimple} in Theorem \ref{theorem2} by analyzing the structure of the optimal solution. 

\begin{theorem}\label{theorem2}
The following assertions are true:
\begin{itemize}
\item[(a)] If $T_i(l)$'s are i.i.d. with a finite mean $\mathbb E[T_i(l)]$ for each $l=1, 2, \ldots, B$, then there exists a function $h(\cdot)$ such that for all $\delta \in \mathbb Z^{+}$ and $0 \leq d \leq B$, the optimal objective value $\bar p_{opt}$ of \eqref{scheduling_problem} satisfies the following Bellman equation: 
\begin{align}\label{Bellman2}
    &h(\delta, d)=\nonumber\\
    &\min_{\substack{l \in \mathbb Z\\1 \leq l \leq B}}\!\!\bigg\{\mathbb E\left[\sum_{k=0}^{Z_l(\delta, d)+T_{1}(l)-1}\!\!\!\! \bigg(\mathrm{err}_{\mathrm{inference}}(\delta+k, d)\!-\! \bar p_{opt}\bigg)\right]\!\!\nonumber\\
&\quad\quad\quad\quad+\!\!\min_{\substack{b \in \mathbb Z\\0 \leq b \leq B-l}}\!\!\mathbb E[h(T_{1}(l)+b, l)]\bigg\}, 
\end{align}
where $h(\cdot)$ is called the relative value function and the function $Z_l(\delta, d)$ is given by 
\begin{align}\label{stoppingtimesoln}
    Z_l(\delta, d)=\min_{\tau \in \mathbb Z}\{\tau \geq 0: \gamma_l(\delta+\tau, d) \geq \bar p_{opt}\},
\end{align}
and the function $\gamma_l(\delta, d)$ is defined in \eqref{gittins}.

\item[(b)] In addition, there exists an optimal solution $\pi^*=((S^*_1, b^*_1, l^*_1), (S^*_2, b^*_2, l^*_2), \ldots) \in \Pi$ to \eqref{scheduling_problem} that is determined by
\begin{align}\label{optimal_featurelength}
&l^*_{i+1}=\nonumber\\
&\argmin_{\substack{l \in \mathbb Z\\1 \leq l \leq B}}\!\!\bigg\{\mathbb E\bigg[\!\!\sum_{k=0}^{\substack{Z_l(T_1(l^*_i)+b^*_i, l^*_i)\\+T_{1}(l)-1}}\!\! \!\!\!\!\!\bigg(\mathrm{err}_{\mathrm{inference}}(\Delta(D_i)+k, l^*_i)\nonumber\\
&\quad\quad\quad- \bar p_{opt}\bigg)\bigg] +\min_{\substack{b \in \mathbb Z\\0 \leq b \leq B-l}} \mathbb E[h(T_{1}(l)+b, l)]\bigg\}, \\ \label{optimal_buffer2}
&b^*_{i+1}=\argmin_{b\in \mathbb Z: 0 \leq b \leq B-l^*_{i+1}} \mathbb E[h(T_{1}(l^*_{i+1})+b, l^*_{i+1})],
\\ \label{OptimalWaitingTime1}
&S^*_{i+1} = \min_{t\in \mathbb Z}\{t \geq D_i: \gamma_{l^*_{i+1}}(\Delta(t), l^*_i) \geq \bar p_{opt}\},
\end{align}
where $\Delta(t)=t-S^*_i+b^*_i$ is the AoI at time $t$ and $D_i=S_i^*+T_i(l_i^*)$ is the $i$-th feature delivery time.
\end{itemize}
\end{theorem}
Theorem \ref{theorem2}(a) simplifies the Bellman equation \eqref{optimal_objective2notsimple} to \eqref{Bellman2}. Unlike \eqref{optimal_objective2notsimple}, which involves joint optimization of three variables, \eqref{Bellman2} is an integer optimization problem. This simplification is possible because, for a given feature length $l$, the original equation \eqref{optimal_objective2notsimple} can be separated into two separated optimization problems. The first problem involves finding the optimal stopping time, denoted by $Z_l(\delta, d)$ defined in \eqref{stoppingtimesoln}, and the second problem is to determine the feature position $b$ that minimizes $\mathbb E[h(T_{1}(l)+b, l)]$. By breaking down the original equation in this way, we can solve the problem more efficiently. Detailed proof of Theorem \ref{theorem2} can be found \ifreport in Appendix D of the supplementary material. \else in our technical report \cite{technical_report}. \fi

Furthermore, Theorem \ref{theorem2}(a) provides additional insights into the solution of \eqref{optimal_objective2notsimple}. Theorem \ref{theorem2}(a) implies that the optimal stopping time $Z^*(\delta, d)$ in \eqref{optimal_objective2notsimple} follows a threshold policy. Specifically, if $l^*(\delta, d)=l$, then $Z^*(\delta, d)$ equals $Z_l(\delta, d)$, which is defined in \eqref{stoppingtimesoln}. Here, $Z_l(\delta, d)$ is the minimum positive integer value $\tau$ for which $\gamma_l(\delta+\tau, d)$ defined in \eqref{gittins} exceeds the optimal objective value $\bar p_{opt}$.

Theorem \ref{theorem2}(b) provides an optimal solution $\pi^* \in \Pi$ to \eqref{scheduling_problem}. According to Theorem \ref{theorem2}(b), by using precomputed $\bar p_{opt}$ and the relative value function $h(\cdot)$, we can obtain the optimal feature length $l^*_{i+1}$ from \eqref{optimal_featurelength} using an exhaustive search algorithm.  After obtaining $l^*_{i+1}$, the optimal feature position $b^*_{i+1}$ can be determined from \eqref{optimal_buffer2}. The optimal scheduling time $S^*_{i+1}$ provided in \eqref{OptimalWaitingTime1} follows a threshold policy. Specifically, the $(i+1)$-th feature is transmitted in time-slot $t$ if two conditions are satisfied: (i) the previous feature is delivered by time $t$, and (ii) the function $\gamma_{l^*_{i+1}}(\Delta(t), l^*_i)$ exceeds the optimal objective value $\bar p_{opt}$ of \eqref{scheduling_problem}.

\begin{algorithm}[t]
\caption{Policy Evaluation Algorithm}\label{alg:PolicyEvaluation}
\begin{algorithmic}[1]
\State Input: $Z_{\pi}(\delta, d)$, $l_{\pi}(\delta, d)$, and $b_{\pi}(\delta, d)$ for all $(\delta, d)$.

\State Initialize $h_{\pi}(\delta, d)$ arbitrarily for all $(\delta, d)$, except for one  fixed state $(\delta', d')$ with $h_{\pi}(\delta', d')=0$.
\State Initialize a small positive number $\alpha_1$ as a threshold. 
\Repeat
   \State $\theta_1 \leftarrow 0$.
   %\State $\tau \leftarrow Z_{\pi}(\delta', d')+T_{1}(l_{\pi}(\delta', d'))$.
  \State Determine $\bar p_{\pi}$ using \eqref{barppi}.
    \For{\text{each state $(\delta, d)$}} 
       \State $\tau' \leftarrow {\red Z_{\pi}(\delta, d)+T_{1}(l_{\pi}(\delta, d))}$.
       \State 
       $h'_{\pi}(\delta, d) \leftarrow \mathbb E[\sum_{k=0}^{\tau'-1} (\mathrm{err}_{\mathrm{inference}}(\delta+k, d)\!-\! \bar p_{\pi})]$
       $~~~~~~~~~+ \mathbb E[h_{\pi}(T_{1}(l_{\pi}(\delta, d))+b_{\pi}(\delta, d),l_{\pi}(\delta, d))].$
        \State $\theta_1 \gets \max\{\theta_1, \big|h'_{\pi}(\delta, d)-h_{\pi}(\delta, d)\big|\}$.
      \EndFor
 \State $h_{\pi} \leftarrow h'_{\pi}$.
 \Until{$\theta_1 \leq \alpha_1$.}
  \State \textbf{return} $\bar p_{\pi}$ and $h_{\pi}(\cdot)$.
\end{algorithmic}
\end{algorithm}

\subsubsection{Policy Iteration Algorithm for Computing $\bar p_{opt}$ and $h(\cdot)$}
To effectively implement the optimal solution $\pi^*\in \Pi$ for \eqref{scheduling_problem}, as outlined in Theorem \ref{theorem2}, it is necessary to precompute the optimal objective value $\bar p_{opt}$ and the relative value function $h(\cdot)$ that satisfies the Bellman equation \eqref{Bellman2}. The computation of $\bar p_{opt}$ and $h(\cdot)$ can be achieved by employing policy iteration algorithm or relative value iteration algorithm for SMDPs, as detailed in \cite[Section 11.4.4]{puterman2014markov}. To apply the relative value iteration algorithm, we need to transform the SMDP into an equivalent MDP. However, this transformation process can be challenging to execute. Therefore, in this paper, we opt to utilize the policy iteration algorithm specifically tailored for SMDPs \cite[Section 11.4.4]{puterman2014markov}. Algorithm \ref{alg:PolicyIteration} provides a policy iteration algorithm for obtaining $\bar p_{opt}$ and $h(\cdot)$, which is composed of two steps: (i) \emph{policy evaluation} and (ii) \emph{policy improvement}. 
 
\emph{Policy Evaluation}: Let $h_{\pi}$ and $\bar p_{\pi}$ be the relative value function and the average inference error under policy $\pi$. Let $l_{\pi}(\delta, d)$, $b_{\pi}(\delta, d)$, and $Z_{\pi}(\delta, d)$ represent feature length, feature position, and waiting time for sending the $(i+1)$-th feature under policy $\pi$ when $\Delta(D_i)=\delta$ and $d(D_i)=d$. Given $l_{\pi}(\delta, d)$, $b_{\pi}(\delta, d)$, and $Z_{\pi}(\delta, d)$ for all $(\delta, d)$, we can evaluate the relative value function $h_{\pi}(\cdot)$ and the average inference error $\bar p_{\pi}$ using Algorithm \ref{alg:PolicyEvaluation}. The relative value function $h_{\pi}(\delta, d)$ represents relative value associated with a reference state. We can set $(\delta', d')$ as a reference state with $h_{\pi}(\delta', d')=0$.
 By using $h_{\pi}(\delta', d')=0$, the average inference error $\bar p_{\pi}$ is determined by
\begin{align}\label{barppi}
  \bar p_{\pi} =& \frac{1}{{\red \mathbb E[\tau]}}\bigg(\mathbb E\left[\sum_{k=0}^{\tau-1} \mathrm{err}_{\mathrm{inference}}(\delta'+k, d') \right]\nonumber\\
  &+ \mathbb E[h_{\pi}(T_{1}(l_{\pi}(\delta', d'))+b_{\pi}(\delta', d'),l_{\pi}(\delta', d'))]\bigg),
\end{align}
where $\tau= {\red Z_{\pi}(\delta', d')+T_{1}(l_{\pi}(\delta', d'))}$. We then use an iterative procedure within Algorithm \ref{alg:PolicyEvaluation} to determine the relative value function $h_{\pi}(\cdot)$. 

 \emph{Policy Improvement}: After obtaining $h_{\pi}$ and $\bar p_{\pi}$ from Algorithm \ref{alg:PolicyEvaluation}, we apply Theorem \ref{theorem2} to derive an improved policy $\pi'$ in Algorithm \ref{alg:PolicyIteration}. Feature length $l_{\pi'}(\delta, d)$, feature position $b_{\pi'}(\delta, d)$, and waiting time $Z_{\pi'}(\delta, d)$ under policy $\pi'$ is determined by
\begin{align}\label{algeq1}
l_{\pi'}(\delta, d)=\argmin\limits_{1 \leq l \leq B}\bigg\{&\mathbb E\left[\sum_{k=0}^{Z_l(\delta, d)+T_{1}(l)-1}\!\! \mathrm{err}_{\mathrm{inference}}(\delta+k, d)\right]\nonumber\\
&-\mathbb E[Z_l(\delta, d)+T_{1}(l)]\bar p_{\pi}\nonumber\\
&+\min\limits_{0 \leq b \leq B-l} \mathbb E[h_{\pi}(T_{1}(l)+b, l)]\bigg\},
\end{align}
\begin{align} \label{algeq2}
\!\!\!\!b_{\pi'}(\delta, d)\!=\!\!\argmin\limits_{0 \leq b \leq B-l_{\pi'}(\delta, d)}\!\!\!\! \mathbb E[h_{\pi}( T_{1}(l_{\pi^{'}}(\delta, d))\!+\!b, l_{\pi'}(\delta, d))],
\end{align}
\begin{align}\label{algeq3}
\!Z_{\pi'}(\delta, d)\!=\!\! \min\limits_{\tau \in \{0, 1, \ldots\}}\big\{ \tau \geq 0: \gamma_{l_{\pi'}(\delta, d)}(\delta+\tau, d) \geq \bar p_{\pi} \big\}.
\end{align}
Instead of a joint optimization problem \eqref{optimal_objective2notsimple}, Algorithm \ref{alg:PolicyIteration} utilizes separated optimization problems \eqref{algeq1}-\eqref{algeq3} based on Theorem \ref{theorem2}. If the improved policy $\pi'$ is equal to the old policy $\pi$, then the policy iteration algorithm converges.  Theorem 11.4.6 in \cite{puterman2014markov} establishes the finite convergence of the policy iteration algorithm of an average cost SMDP. 

\begin{algorithm}[t]
\caption{Policy Iteration Algorithm}\label{alg:PolicyIteration}
\begin{algorithmic}[1]
\State Initialize $Z_{\pi}(\delta, d)$, $l_{\pi}(\delta, d)$, and $b_{\pi}(\delta, d)$ for all $(\delta, d)$.
\State Initialize a small positive number $\alpha_2$ as threshold.
\Repeat
\State $\theta_2 \gets 0$.
\State Obtain $h_{\pi}(\cdot)$ and $\bar p_{\pi}$ from Algorithm \ref{alg:PolicyEvaluation}.
\For {all $(\delta, d)$} 
\State Get $l_{\pi'}(\delta, d)$, $b_{\pi'}(\delta, d)$, $Z_{\pi'}(\delta, d)$ using \eqref{algeq1}-\eqref{algeq3}.
%\State Compute $b_{\pi'}(\delta, d)$ using \eqref{alg2}.
%\State Compute $Z_{\pi'}(\delta, d)$ using \eqref{alg3}.
\State $\theta_2 \gets \max\bigg\{\theta_2, |l_{\pi'}(\delta, d)-l_{\pi}(\delta, d)|$\\
$~~~~~~~~~~~~~+|b_{\pi'}(\delta, d)-b_{\pi}(\delta, d)|+|Z_{\pi'}(\delta, d)-Z_{\pi}(\delta, d)|\bigg\}.$
\State $l_{\pi}(\delta, d) \gets l_{\pi'}(\delta, d).$
\State $b_{\pi}(\delta, d) \gets b_{\pi'}(\delta, d).$
\State $Z_{\pi}(\delta, d) \gets Z_{\pi'}(\delta, d).$
\EndFor
\Until{$\theta_2 \leq \alpha_2$.}
\State \textbf{return} $\bar p_{opt}\gets \bar p_{\pi}$ and $h\gets h_{\pi}$.
\end{algorithmic}
\end{algorithm}

{Now, we discuss the time-complexity of Algorithms \ref{alg:PolicyEvaluation}-\ref{alg:PolicyIteration}. 
To manage the infinite set of AoI values in practice, we introduce an upper bound denoted as $\delta_{\mathrm{bound}}$. Whenever $\delta$ exceeds $\delta_{\mathrm{bound}}$, we set $h_{\pi}(\delta, d) = h_{\pi}(\delta_{\mathrm{bound}}, d)$ for all $d$. Hence, each iteration of our policy evaluation step requires one pass through the approximated state space $\{1, 2, \ldots, \delta_{\mathrm{bound}}\} \times \{1, 2, \ldots, B\}$. Therefore, the time complexity of each iteration is $O(\delta_{\mathrm{bound}}B)$, assuming that the required expected values are precomputed. Considering the bounded set $\{0, 1, \ldots, \delta_{\mathrm{bound}}\}$ instead of $\mathbb{Z}^{+}$, the time complexities of \eqref{algeq1}, \eqref{algeq2}, and \eqref{algeq3} are $O(B^2)$, $O(B)$, and $O(\delta_{\mathrm{bound}})$, respectively, provided that the expected values in \eqref{algeq1}-\eqref{algeq3} are precomputed. The overall complexity of \eqref{algeq1}-\eqref{algeq3} is $O(\max\{B^2, B, \delta_{\mathrm{bound}}\})$, which is more efficient than the joint optimization problem \eqref{optimal_objective2notsimple}. The latter has a time complexity of $O(\delta_{\mathrm{bound}}B^2)$.
In each iteration of the policy improvement step, the optimization problems \eqref{algeq1}-\eqref{algeq3} are solved for all state $(\delta, d)$ such that $\delta=1,2, \ldots, \delta_{\mathrm{bound}}$ and $d=1, 2, \ldots, B$. Hence, the total complexity of each iteration of the policy improvement step is $O(\max\{B^3\delta_{\mathrm{bound}},B\delta_{\mathrm{bound}}^2\})$.}

\section{Learning and Communications Co-design: Multiple Source Case}

\subsection{System Model}
\begin{figure}[t]
\centering
\includegraphics[width=0.40\textwidth]{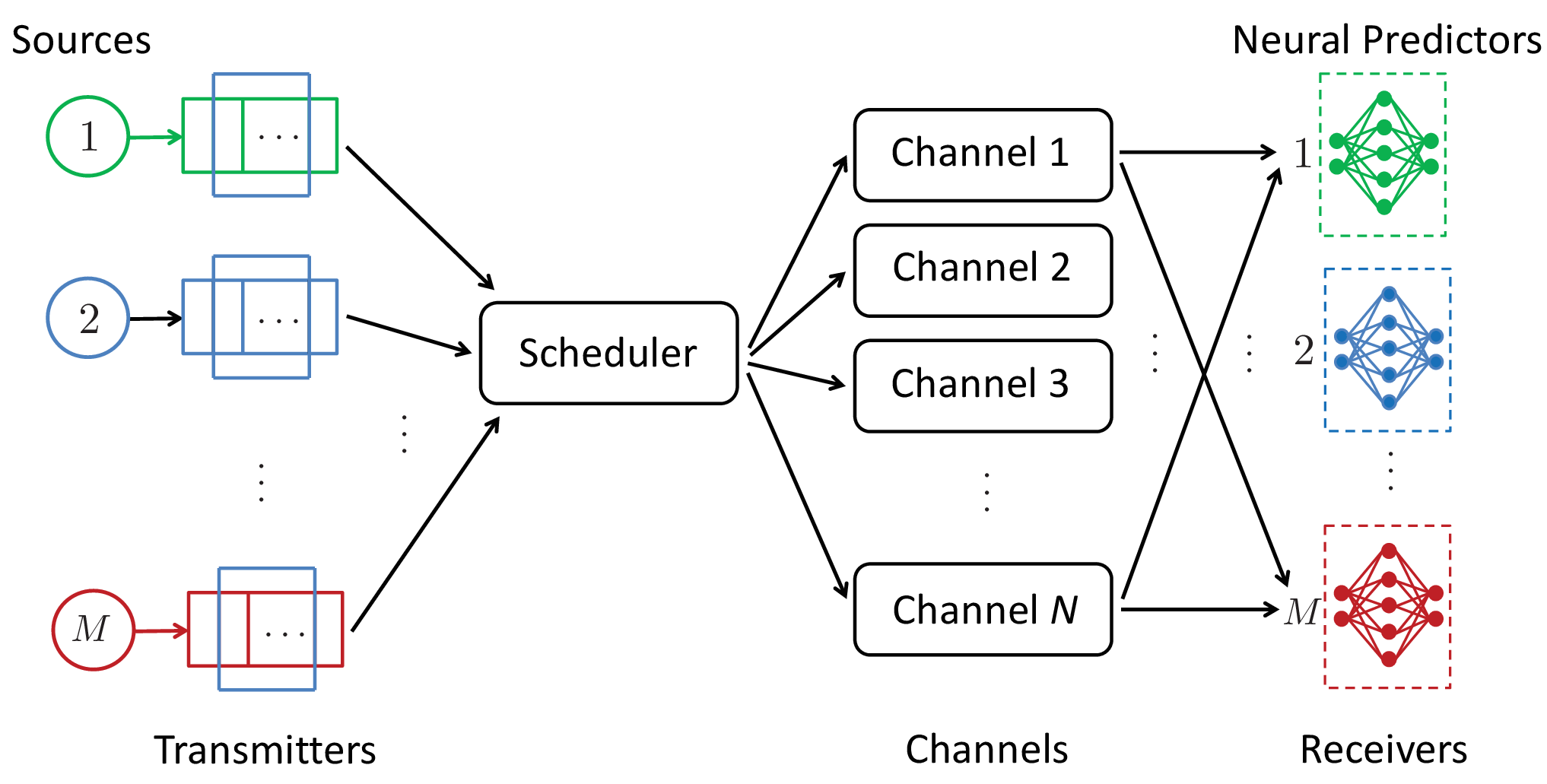}
\caption{\small  A multiple source-predictor pairs and multiple channel remote inference system. \label{fig:mult-scheduling}
}
\vspace{-3mm}
\end{figure}
Consider a remote inference system consisting of $M\geq 1$ source-predictor pairs connected through $N \geq 1$ shared communication channels, as illustrated in Fig. \ref{fig:mult-scheduling}. Each source $j$ has a buffer that stores $B_j$ most recent signal observations at each time slot $t$. At time slot $t$, a centralized scheduler determines whether to send a feature from source $j$ with feature length $l_j(t)$ and feature position $b_j(t)$. We denote $l_j(t)=0$ if the scheduler decides not to send a feature from source $j$ at time $t$. If a feature from source $j$ is sent, we assume it will be delivered to the $j$-th neural predictor in the next time slot using $l_j(t)$ channel resources. The transmission model of the multiple source system is significantly different from that of the single source model discussed in Section \ref{TransmissionPolicy}. In the latter case, only one channel was considered, while $N$ communication channels are available in the former. The channels could be from multiple frequencies and/or time resources. For example, if the clock rate in the multiple access control (MAC) layer is faster than that of the application layer, then one application-layer time-slot could comprise multiple MAC-layer time-slots. A feature can utilize multiple channels (i.e., frequency or time resources) for transmission during a single time slot. However, the channel resource is limited, so the system must satisfy
\begin{align}
\sum_{j=1}^M l_j(t) \leq N.
\end{align}

The system begins operating at time $t=0$. Let $S_{j, i}$ denote the sending time of the $i$-th feature from the $j$-th source. Since we assume that a feature takes one time-slot to transmit, the corresponding neural predictor receives the $i$-th feature from the $j$-th source at time $S_{j, i}+1$. The AoI of the source $j$ at time slot $t$ is defined as
\begin{align}\label{Def_AoI_multiple}
\Delta_j(t):=t-S_{j, i}+b_j(S_{j, i}),~\text{if}~ S_{j,i} < t \leq S_{j, i+1}.
\end{align}
We denote $d_j(t)$ as the feature length of the most recent received feature from $j$-th source by time $t$, given by
\begin{align}\label{Def_featurelength_multiple}
d_j(t)=l_j(S_{j, i}),~\text{if}~ S_{j,i} < t \leq S_{j, i+1}.
\end{align}

\subsection{Scheduling Policy}
At time slot $t$, a centralized scheduler determines the value of the feature length $l_j(t)$ and the feature position $b_j(t)$ for every $j$-th source. A scheduling policy is denoted by $\pi={(\pi_{j})}_{j=1}^M$, where $\pi_{j}=((l_j(1), b_j(1)), (l_j(2), b_j(2)), \ldots)$. Let $\Pi$ denote the set of all the causal scheduling policies that determine $l_j(t)$ and $b_j(t)$ based on the current and the historical information available at the transmitter such that $0 \leq l_j(t)+b_j(t) \leq B_j$.

\subsection{Problem Formulation}
Our goal is to minimize the time-averaged sum of the inference errors of the $M$ sources, which is formulated as
\begin{align}\label{Multi-scheduling_problem}
&\inf_{ \pi \in \Pi}\sum_{j=1}^M  \limsup_{T\rightarrow \infty}\mathbb{E}_{\pi} \left[\frac{1}{T} \sum_{t=0}^{T-1} p_j(\Delta_j(t), d_j(t))\right], \\\label{Sceduling_constraint}
&\ \mathrm{s.t.}\sum_{j=1}^M l_j(t) \leq N, ~~t=0, 1, 2, \ldots,
\end{align}
where $p_j(\Delta_j(t), d_j(t))$ is the inference error of source $j$ at time slot
$t$.

The problem \eqref{Multi-scheduling_problem}-\eqref{Sceduling_constraint} can be cast into an infinite-horizon average cost restless multi-armed multi-action bandit problem \cite{whittle1988restless, hodge2015asymptotic} by viewing each source $j$ as an arm, where a scheduler needs to decide multiple actions $(l_j(t), b_j(t))_{j=1}^M$ at every time $t$ by observing state $(\Delta_j(t), d_j(t))$.

Finding an optimal solution to the RMAB problem is PSPACE hard \cite{papadimitriou1994complexity}. Whittle, in his seminal work \cite{whittle1988restless}, proposed a heuristic policy for RMAB problem with binary action. In \cite{hodge2015asymptotic}, a modified Whittle index policy is proposed for the multi-action RMAB problems. Whittle index policy is known to be asymptotically optimal \cite{weber1990index}, but the policy needs to satisfy a complicated indexability condition. Proving indexability is challenging for our multi-action RMAB problem because we allow (i) general penalty function $p_j(\delta, l)$ that is not necessarily monotonic with respect to AoI $\delta$ and (ii) time-variant feature length. To this end, we propose a low-complexity algorithm that does not need to satisfy any indexability condition.

\subsection{Lagrangian Optimization of a Relaxed Problem}
Similar to Whittle's approach \cite{whittle1988restless}, we utilize a Lagrange relaxation of the problem \eqref{Multi-scheduling_problem}-\eqref{Sceduling_constraint}. We first relax the per time-slot channel constraint \eqref{Sceduling_constraint} as the following time-average expected channel constraint
\begin{align}\label{Changed_constraint}
\sum_{j=1}^M \limsup_{T \to \infty}  \mathbb E_{\pi}\left[\frac{1}{T}\sum_{t=0}^{T-1}  l_j(t)\right] \leq N. 
\end{align}
The relaxed constraint \eqref{Changed_constraint} only needs to be satisfied on average, whereas \eqref{Sceduling_constraint} is required to hold at every time-slot. By this, the original problem \eqref{Multi-scheduling_problem}-\eqref{Sceduling_constraint} becomes 
\begin{align}\label{Multi-scheduling_problem1}
&\inf_{ \pi \in \Pi} \sum_{j=1}^M  \limsup_{T\rightarrow \infty}\mathbb{E}_{\pi} \left[\frac{1}{T} \sum_{t=0}^{T-1} p_j(\Delta_j(t), d_j(t))\right],  \\ \label{Changed_constraint1}
&~\mathrm{s.t.}\sum_{j=1}^M \limsup_{T \to \infty}  \mathbb E_{\pi}\left[\frac{1}{T}\sum_{t=0}^{T-1}  l_j(t)\right] \leq N.
\end{align}
{The relaxed problem \eqref{Multi-scheduling_problem1}-\eqref{Changed_constraint1} is of interest as the optimal solution of the problem provides a lower bound to the original problem \eqref{Multi-scheduling_problem}-\eqref{Sceduling_constraint}. 

\subsubsection{Lagrangian Dual Decomposition of \eqref{Multi-scheduling_problem1}-\eqref{Changed_constraint1}}
To solve \eqref{Multi-scheduling_problem1}-\eqref{Changed_constraint1}, we utilize a Lagrangian dual decomposition method \cite{whittle1988restless, palomar2006tutorial}. At first, we apply Lagrangian multiplier $\lambda\geq 0$ to the time-average channel constraint \eqref{Changed_constraint1} and get the following Lagrangian dual function
\begin{align}\label{Multi-scheduling_problem_Lagrange}
&q(\lambda)\nonumber\\
=& \inf_{ \pi \in \Pi}\sum_{j=1}^M  \limsup_{T\rightarrow \infty} \mathbb{E}_{\pi}\left [ \frac{1}{T} \sum_{t=0}^{T-1} \bigg(p_j(\Delta_j(t), d_j(t))\!+\! \lambda l_j(t)\bigg)\right]\nonumber\\
&-\lambda N.
\end{align}
The problem \eqref{Multi-scheduling_problem_Lagrange} can be decomposed into $M$ sub-problems. The sub-problem associated with the $j$-th source is defined as:
%\begin{align}
%p_{\pi_j}(\lambda)=\limsup_{T\rightarrow \infty} \mathbb{E}_{\pi_{j}}\left [ \frac{1}{T} \sum_{t=0}^{T-1} \bigg(p_j(\Delta_j(t), d_j(t))\!+\! \lambda l_j(t)\bigg)\right].
%\end{align}
%For a given $\lambda$, the inner optimization of the Lagrangian problem \eqref{Multi-scheduling_problem_Lagrange} is 
%\begin{align}\label{Multi-scheduling_problem_Lagrange_inner}
%\inf_{ \pi \in \Pi}\sum_{j=1}^M  p_{\pi_j}(\lambda)-\lambda N,
%\end{align}
%which can be decomposed into $M$ sub-problems. The sub-problem associated with $j$-th source is given by
\begin{align}\label{decoupled_problem}
\bar p_{j}(\lambda)\!\!=\!\!\!\inf_{\pi_{j} \in \Pi_{j}} \!\!\limsup_{T\rightarrow \infty}\frac{1}{T} \mathbb{E}_{\pi_{j}}\!\left [ \sum_{t=0}^{T-1} \bigg(p_j(\Delta_j(t), d_j(t))\!+\! \lambda l_j(t)\bigg)\right],
\end{align}
where $\Pi_{j}$ is the set of all causal scheduling policies $\pi_j$. The sub-problem \eqref{decoupled_problem} is an infinite horizon average cost MDP, where a scheduler decides action $(l_j(t), b_j(t))$ by observing state $(\Delta_j(t), d_j(t))$. The Lagrange multiplier $\lambda$ in \eqref{decoupled_problem} can be interpreted as a transmission cost: whenever $l_j(t)=l$, the source $j$ has to pay cost of $\lambda l$ for using $l$ channel resources. 

The optimal solution to \eqref{decoupled_problem} can be obtained by solving the following Bellman equation:
\begin{align}\label{BellmanMDP}
h_{j, \lambda}(\delta, d) &= \min_{\substack{l \in \mathbb{Z}, b \in \mathbb{Z}\\0 \leq l + b \leq B_j}} Q_{j, \lambda}((\delta, d), (l, b)),
\end{align}
where $h_{j, \lambda}(\cdot)$ represents the relative value function of the MDP \eqref{decoupled_problem}, and the function $Q_{j, \lambda}(\cdot,\cdot)$ is defined as follows
\begin{align}\label{actionvalue3}
&Q_{j, \lambda}((\delta, d), (l, b))\nonumber\\
&:=\begin{cases}
p_j(\delta, d) - \bar p_{j}(\lambda) + h_{j, \lambda}(\delta + 1, d), & \text{if}~l = 0,\\
p_j(\delta, d) - \bar p_{j}(\lambda) + h_{j, \lambda}(b + 1, l) + \lambda l, & \text{otherwise}.
\end{cases}
\end{align}
The relative value function $h_{j, \lambda}(\cdot)$ can be computed using the relative value iteration algorithm \cite{puterman2014markov, bertsekasdynamic}.

Let $\pi_{j, \lambda}^*=((l^*_{j, \lambda}(1), b^*_{j, \lambda}(1)), (l^*_{j, \lambda}(2), b^*_{j, \lambda}(2)), \ldots)$ be an optimal solution to \eqref{decoupled_problem}, which is derived by using \eqref{BellmanMDP} and \eqref{actionvalue3}. The optimal feature length $l^*_{j, \lambda}(t)$ is determined by
\begin{align}\label{optimallengtht}
&\!\!\!\!\!l^*_{j, \lambda}(t)\nonumber\\
&\!\!\!\!\!= \argmax\limits_{l \in \mathbb Z:0\leq l \leq B_j} 
h_{j, \lambda}(\Delta(t)+1, d(t))\!\!-\!h_{j, \lambda}(\hat b_{j, \lambda}(l)+1, l)\!\!-\!\lambda l,
\end{align}
where the function $\hat b_{j, \lambda}(l)$ is given by
\begin{align}\label{featurefunction}
\hat b_{j, \lambda}(l)=\argmin_{b \in \mathbb Z: 0 \leq b \leq B_j-l}h_{j, \lambda}(b+1, l),
\end{align}
The optimal feature position in $\pi_{j, \lambda}^*$ is 
\begin{align}\label{OptimalPositiont}
b^*_{j, \lambda}(t)=\hat b_{j, \lambda}(l_{j, \lambda}(t)).
\end{align}

\subsubsection{Lagrange Dual Problem} 
Next, we determine the optimal dual cost $\lambda^*$ that solves the following Lagrange dual problem:
\begin{align}\label{DualProblem}
\max_{\lambda \geq 0} q(\lambda),
\end{align}
where $q(\lambda)$ is the Lagrangian dual function defined in \eqref{Multi-scheduling_problem_Lagrange}.
To get $\lambda^*$, we apply the stochastic sub-gradient ascent method \cite{palomar2006tutorial}, which iteratively updates $\lambda(k)$ as follows
\begin{align}\label{updatelambda}
\!\!\!\!\!\lambda(k+1)\!=\! \max\bigg\{\lambda(k)+\frac{\beta}{k} \left(\sum_{j=1}^M l_{j, \lambda(k)}(k)-N\right), 0\bigg\},
\end{align}
where $k$ is the iteration index, $\beta>0$ determines the step size $\frac{\beta}{k}$, and $l_{j, \lambda(k)}(k)$ is the feature length of source $j$ at the $k$-th iteration. Detailed optimization technique is provided in Algorithm \ref{alg:lagrange}. 
\begin{algorithm}[t]
\caption{Dual Algorithm to Solve \eqref{DualProblem}}\label{alg:lagrange}
{\begin{algorithmic}[1]
\State Input: Step size $\beta>0$ and dual cost $\lambda(1)=0$.
 \State Initialize $\Delta_j(0)$, $d_j(0)$, $l_j(0)$, and $b_j(0)$ for all $j$.
\State Initialize a small positive number $\theta$ as threshold. 
\Repeat
\For{\text{each source $j$}}
\If{$l_j(k-1)>0$}
    \State $\Delta_j(k) \leftarrow 1+b_j(k-1)$, $d_j(k) \gets l_j(k-1)$.
  \Else{}
    \State $\Delta_j(k) \leftarrow \Delta_j(k-1)+1$, $d_j(k) \leftarrow d_j(k-1)$.
    \EndIf
    %\State Update $\alpha_{j, \lambda(k)}(\Delta_j(k), d_j(k), l)$ using \eqref{netgain3}.
\State Compute $l_{j, \lambda(k)}(k)$ using \eqref{optimallengtht}.

\State Compute $b_{j, \lambda(k)}(k)$ using \eqref{OptimalPositiont}.
    \EndFor
\State Update $\lambda(k+1)$ using \eqref{updatelambda}.
 \Until{$|\lambda(k+1)-\lambda(k)|\leq \theta$.}
 \State \textbf{return} $\lambda^* \leftarrow \lambda(k+1)$
\end{algorithmic}}
\end{algorithm}

\subsection{Net Gain Maximization Policy} 
After getting optimal dual cost $\lambda^*$, we can use policy $(\pi_{j, \lambda^*})_{j=1}^M$ for the relaxed problem \eqref{Multi-scheduling_problem1}-\eqref{Changed_constraint1}. But it is infeasible to implement the policy for the original problem \eqref{Multi-scheduling_problem}-\eqref{Sceduling_constraint} because it may violate the scheduling constraint \eqref{Sceduling_constraint}. Motivated by Whittle’s approach \cite{whittle1988restless}, we aim to select actions with higher priority, while satisfying the scheduling constraint \eqref{Sceduling_constraint} at every time slot.
Towards this end, we introduce ``Net Gain", denoted as $\alpha_{j, \lambda}(\delta, d, l)$, to measure the advantage of selecting feature length $l$, which is given by
\begin{align}\label{netgain1}
&\alpha_{j, \lambda}(\delta, d, l)\nonumber\\
&:=Q_{j, \lambda}((\delta, d), (0, \hat b_{j, \lambda}(l)))-Q_{j, \lambda}((\delta, d), (l, \hat b_{j, \lambda}(l))),
\end{align}
where the function $Q_{j, \lambda}$ is defined in \eqref{actionvalue3} and the function $\hat b_{j, \lambda}$ is defined in \eqref{featurefunction}.  Substituting \eqref{actionvalue3} into \eqref{netgain1}, we get  
\begin{align}\label{netgain3}
\!\!\! \alpha_{j, \lambda}(\delta, d, l)=h_{j, \lambda}(\delta+1, d)-h_{j, \lambda}(\hat b_{j, \lambda}(l)+1, l)- \lambda l.
\end{align} 

For a given $\lambda$, the net gain $\alpha_{j, \lambda}(\delta, d, l)$ has an economic interpretation. Given the state $(\delta, d)$ of source $j$, the net gain $\alpha_{j, \lambda}(\delta, d, l)$ measures the maximum reduction in the loss by selecting source $j$ with feature length $l$, as opposed to not selecting source $j$ at all. If $\alpha_{j, \lambda}(\delta, d, l)$ is negative for all $l=1, 2, \ldots, B_j$, then it better not to select source $j$.
If $\alpha_{j, \lambda}(\Delta_j(t), d_j(t), l_j) >\alpha_{k, \lambda}(\Delta_k(t), d_k(t), l_k)$, then the feature length $l_j$ for source $j$ is prioritized over the feature length $l_k$ for source $k$. Under the constraint \eqref{Sceduling_constraint}, we select feature lengths that maximize ``Net Gain":
\begin{align}\label{problemalgorithm}
&\max_{\substack{0 \leq l_j(t) \leq B\\l_j(t) \in \mathbb Z, \forall j}} \sum_{j=1}^M \alpha_{j, \lambda^*}(\Delta_j(t), d_j(t), l_j(t)), \\ \label{constraintagain}
&\quad~~\mathrm{s.t.} \quad\sum_{j=1}^M l_j(t) \leq N.
\end{align}
The ``Net Gain Maximization" problem \eqref{problemalgorithm} with constraint \eqref{constraintagain} is a bounded Knapsack problem. By using \eqref{problemalgorithm}-\eqref{constraintagain}, we propose a new algorithm for the problem \eqref{Multi-scheduling_problem}-\eqref{Sceduling_constraint} in Algorithm \ref{alg:multischeduling}. 

\begin{algorithm}[t]
\caption{Net Gain Maximization Policy}\label{alg:multischeduling}
\begin{algorithmic}[1]
\State Input: Optimal dual variable $\lambda^*$ obtained in Algorithm \ref{alg:lagrange}.
\State Compute $\alpha_{j, \lambda^*}(\delta, d, l)$ using \eqref{netgain3} for all $j,\delta, d, l$.
\For{each time $t \geq 0$}
\State Update $\Delta_j(t) $ and $ d_j(t)$ using \eqref{Def_AoI_multiple} and \eqref{Def_featurelength_multiple} for all source $j$.
\State Compute $(l_j(t))_{j=1}^M$ by solving problem \eqref{problemalgorithm}-\eqref{constraintagain}.
\State  $(b_j(t))_{j=1}^M \gets (\hat b_{j, \lambda^*}(l_j(t)))_{j=1}^M$ by 
using \eqref{featurefunction}.
\EndFor
\end{algorithmic}
\end{algorithm}

{Algorithm \ref{alg:multischeduling} starts from $t=0$. At time $t=0$, the algorithm takes the dual variable (transmission cost) $\lambda^*$ from Algorithm \ref{alg:lagrange} which is run offline before $t=0$. The ``Net Gain" $\alpha_{j, \lambda^*}(\delta, d, l)$ is precomputed for every source $j$, every feature length $l$, and every state $(\delta, d)$ such that $\delta \in \mathbb Z^{+}$, $l, d \in \{1, 2, \ldots, B_j\}$, where we approximate infinite set of AoI values $\mathbb Z^{+}$ by using an upper bound $\delta_{\mathrm{bound}}$. We can set $\alpha_{j, \lambda^*}(\delta, d, l)=\alpha_{j, \lambda^*}(\delta_{\mathrm{bound}}, d, l)$ if $\delta>\delta_{\mathrm{bound}}$.
 
From time $t\geq 0$, Algorithm \ref{alg:multischeduling} solves the knapsack problem \eqref{problemalgorithm}-\eqref{constraintagain} at every time slot $t$. The knapsack problem is solved by using a dynamic programming method in $O(MNB)$ time \cite{andonov2000unbounded}, where $M$ is the number of sources, $N$ is the number of channels, and $B$ is the maximum buffer size among all source $j$. The feature position $b_j(t)$ is obtained from a look up table that stores the value of function $\hat b_{j, \lambda^*}(l)$ for all $j$ and $l$.}

Unlike the Whittle index policy \cite{whittle1988restless}, our policy proposed in Algorithm \ref{alg:multischeduling} does not need to satisfy any indexability condition. There exists some other policies that do not need to satisfy indexability condition \cite{xiong2022index, chen2021scheduling}. The policies in \cite{xiong2022index, chen2021scheduling} are developed based on linear programming formulations, our policy does not need to solve any linear programming.

\section{Trace-driven Evaluations}\label{numerical}

In this section, we demonstrate the performance of our scheduling policies. The performance evaluation is conducted using an inference error function obtained from a channel state information (CSI) prediction experiment. In Fig. \ref{fig:Trainingcsi}, one can observe the inference error function of a CSI prediction experiment. The discrete-time autocorrelation function of the generated fading channel coefficient is defined as $r(k) = b J_0(2\pi f_dT_s|k|)$, where $r(k)$ represents the autocorrelation of the CSI signal process with time lag $k$, $b$ signifies the variance of the process, $J_0(\cdot)$ denotes the zeroth-order Bessel function, $T_s$ is the channel sampling duration, $f_d = \frac{vf_c}{c}$ is the maximum Doppler shift, $v$ stands for the velocity of the source, $f_c$ is the carrier frequency, and $c$ represents the speed of light. In this experiment, we employed a quadratic loss function. {Although we utilize the CSI prediction experiment and a quadratic loss function for evaluating the performance of our scheduling policies, we note that our scheduling policies are not limited to any specific experiment, loss function, or predictor.}

 \subsection{Single Source Scheduling Policies}
We evaluate the following four single source scheduling policies. 
\begin{itemize}
\item[1.] Generate-at-Will, Zero Wait with Feature Length $l$: In this policy, $S_{i+1}=S_i+T_i(l_i)$, $b_i=0$, and $l_i=l$ for all $i$-th feature transmissions.  
\item[2.] Optimal Policy with Time-invariant Feature Length (TIFL): The policy that we propose in Theorem \ref{theorem1}.
\item[3.] Optimal Policy with Time-variant Feature Length (TVFL): The policy that we propose in Theorem \ref{theorem2}.
\item[4.] Periodic Updating with Feature Length $l$: After every time slot $T_p$, the policy submits features with feature length $l$ and feature position $0$ to a First-Come, First-Served communication channel. 
\end{itemize}

We evaluate the performance of the above four single source scheduling policies, where the task to infer the current CSI of a source by observing features. For generating the CSI dataset, we set $b_0=1$, $T_s=1\text{ms}$, $v=15~\text{m/s}$, and $f_c=2\text{GHz}$. Additionally, we add white noise to the feature variable with a variance of $10^{-6}$. 

In the single source case, we consider that the $i$-th feature requires $T_i(l)=\lceil \alpha l \rceil$ time-slots for transmission, where $\alpha$ represents the communication capacity of the channel. For example, if the number of bits used for representing a CSI symbol is $n$ and the bit rate of the channel is $\rho$, then $\alpha=\frac{\rho}{n}$.

\begin{figure}
\centering
\includegraphics[width=0.30\textwidth]{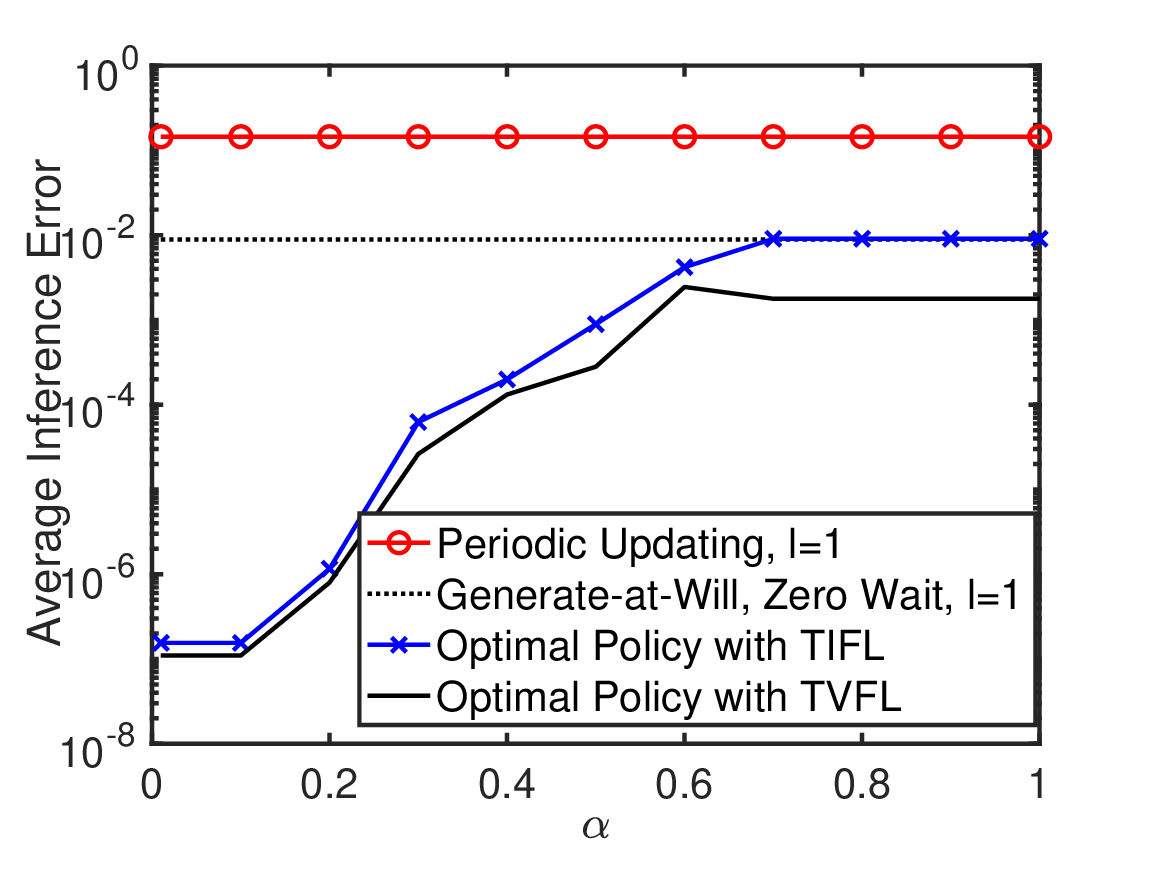}
\caption{\small Single Source Case: Time-averaged inference error vs. the scale parameter $\alpha$ in transmission time $T_i(l)=\lceil \alpha l \rceil$ for all $i$.\label{fig:numericalsingle1}
}
\vspace{-3mm}
\end{figure}

\begin{figure}
\centering
\includegraphics[width=0.30\textwidth]{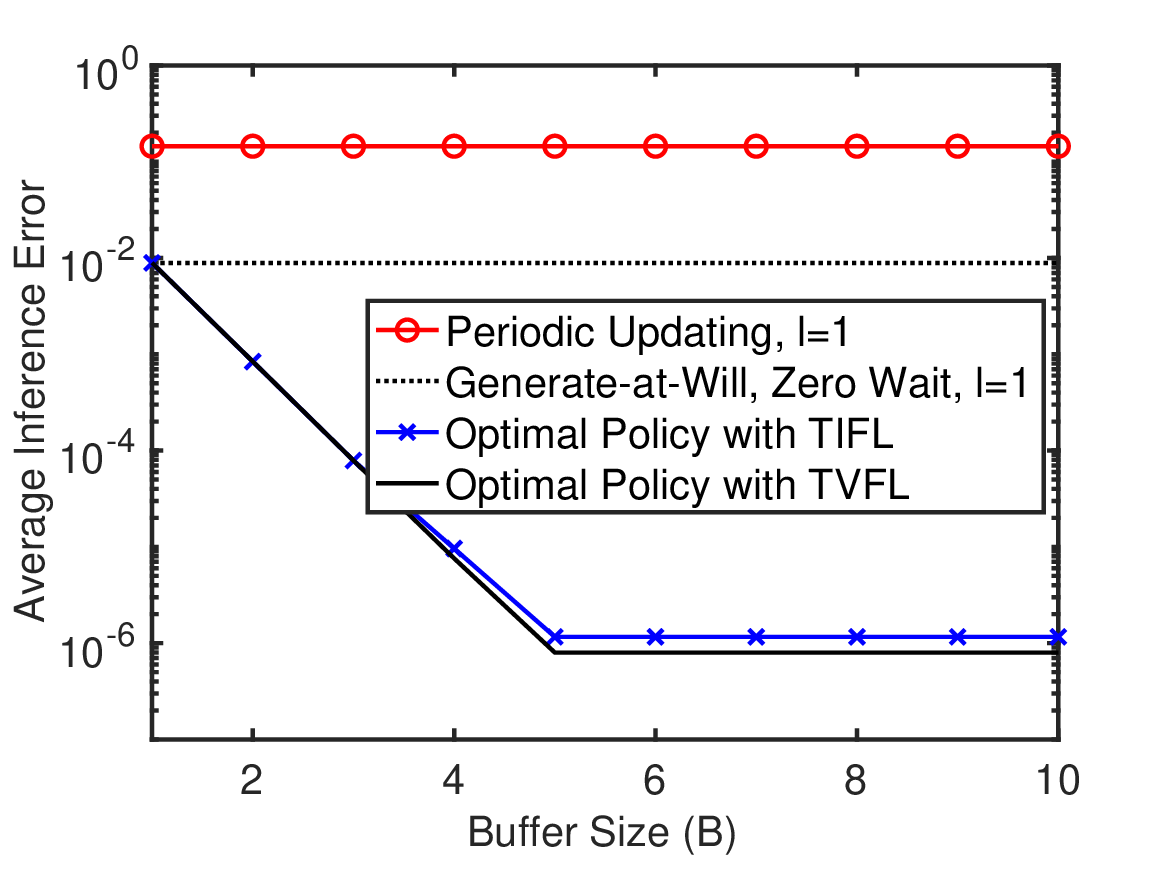}
\caption{\small Single Source Case: Time-averaged inference error vs. the buffer size $B$. \label{fig:numericalsingle2}
}
\vspace{-3mm}
\end{figure}

Fig. \ref{fig:numericalsingle1} shows the time-averaged inference error under different policies against the parameter $\alpha$, where $\alpha>0$. The plot is constrained to $\alpha=1$ since values of $\alpha>1$ is impractical due to the possibility of sending CSI using fewer bits. The buffer size of the source is $B=10$. Among the four scheduling policies, the ``Optimal Policy with TVFL" yields the best performance, while the ``Optimal Policy with TIFL" outperforms the other two policies. The findings in Figure \ref{fig:numericalsingle1} demonstrate that when $\alpha\leq 0.1$, the ``Optimal Policy with TVFL" can achieve a performance improvement of $10^4$ times compared to the ``Periodic Updating, $l=1$" with $T_p=4$ and ``Generate-at-Will, Zero Wait, $l=1$" policies. This result is not surprising since ``Periodic Updating, $l=1$" and ``Generate-at-Will, Zero Wait, $l=1$" do not utilize longer features, despite all features with $l=1, 2, \ldots, 10$ taking only $1$ time slot when $\alpha \leq 0.1$. When $\alpha>0.1$, the average inference error of the ``Periodic Updating" and ``Generate-at-Will, Zero Wait" policies are at least 10 times worse than that of the ``Optimal Policy with TVFL.'' The reasons are as follows: (1) The ``Periodic Updating" policy does not transmit a feature even when the channel is available, leading to an inefficient use of resources. In our simulation, this situation is evident as $T_i(1)=1$ and $T_p=4$. Again, ``Periodic Updating" may transmit features even when the preceding feature has not yet been delivered, resulting in an extended waiting time for the queued feature. This frequently leads to the receiver receiving a feature with a significantly large AoI value, which is not good for accurate inference. (2) Conversely, the ``Generate-at-Will, Zero-Wait" policy isn't superior because zero-wait is not advantageous, and the feature position $b=0$ may not be an optimal choice since the inference error is non monotonic with respect to AoI. 

The policy ``Optimal Policy with TIFL" achieves an average inference error very close to that of the ``Optimal Policy with TVFL,'' but it is simpler to implement. Furthermore, the ``Optimal Policy with TIFL" requires only one predictor associated with the optimal time-invariant feature length and does not require switching the predictor. 
 
%Interestingly, Fig. \ref{fig:numericalsingle1} shows that for both ``Periodic Updating" and ``Generate-at-Will, Zero-Wait" policies, sometimes the feature length $l=1$ can be a better choice than feature length $l=10$. This is because, as the feature length increases, the number of time-slots required for transmission also increases, and eventually, a stale feature is delivered to the receiver, resulting in large inference errors. This observation verifies our motivation for the selection of feature length. 

Fig. \ref{fig:numericalsingle2} plots the time-averaged inference error vs. the buffer size $B$. In this simulation, $\alpha=0.2$ is considered. The results show that increasing $B$ can improve the performance of the ``Optimal Policy with TVFL" and ``Optimal Policy with TIFL" compared to the other policies. As $B$ increases, ``Optimal Policy with TVFL" and ``Optimal Policy with TIFL" outperform the others. In contrast, the ``Periodic Updating" and ``Generate-at-Will" policies do not utilize the buffer and their performance remains unchanged with increasing $B$. Moreover, we can notice that the buffer size $B=5$ is enough for this experiment as further increase in buffer size does not improve the performance. 

\begin{figure}
\centering
\includegraphics[width=0.30\textwidth]{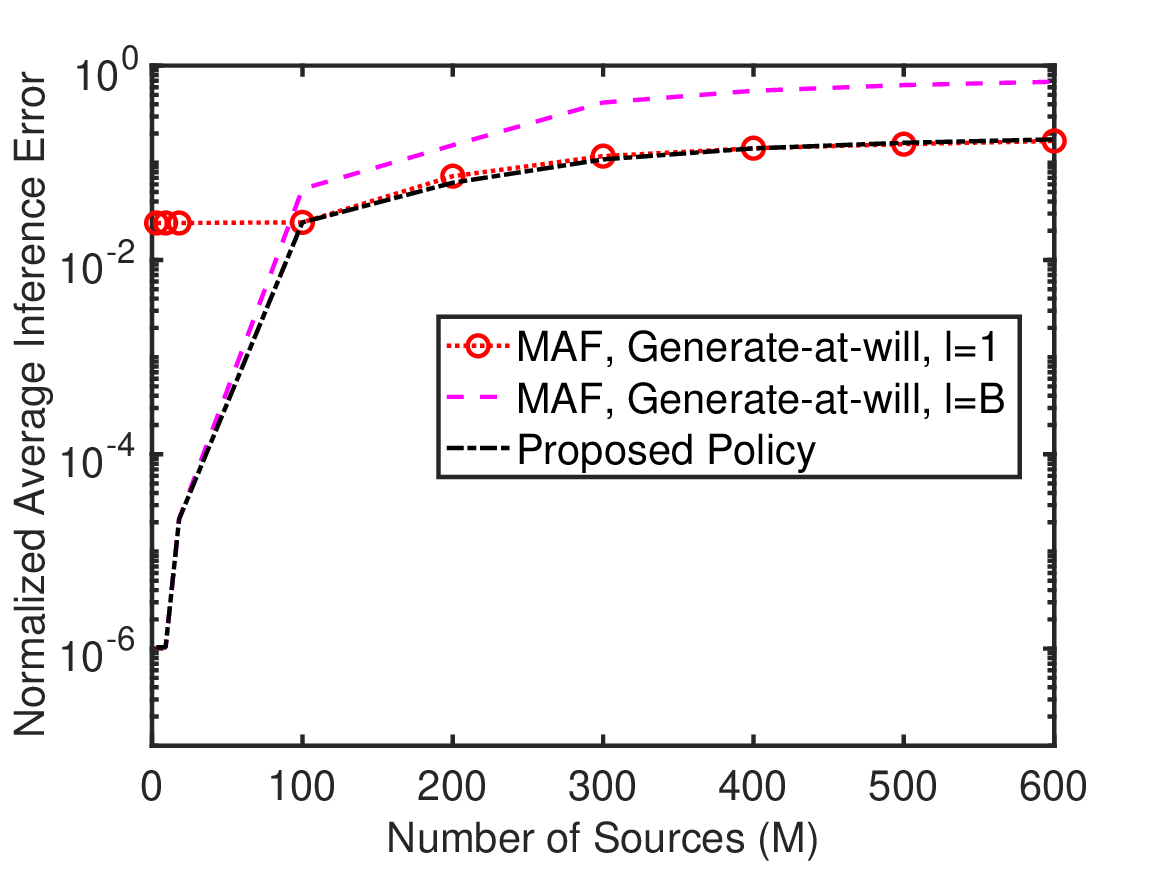}
\caption{\small Multiple Source Case: Time-averaged inference error vs. the number of sources $M$.  \label{fig:numericalmultiple11}
}
\vspace{-3mm}
\end{figure}

\begin{figure}
\centering
\includegraphics[width=0.30\textwidth]{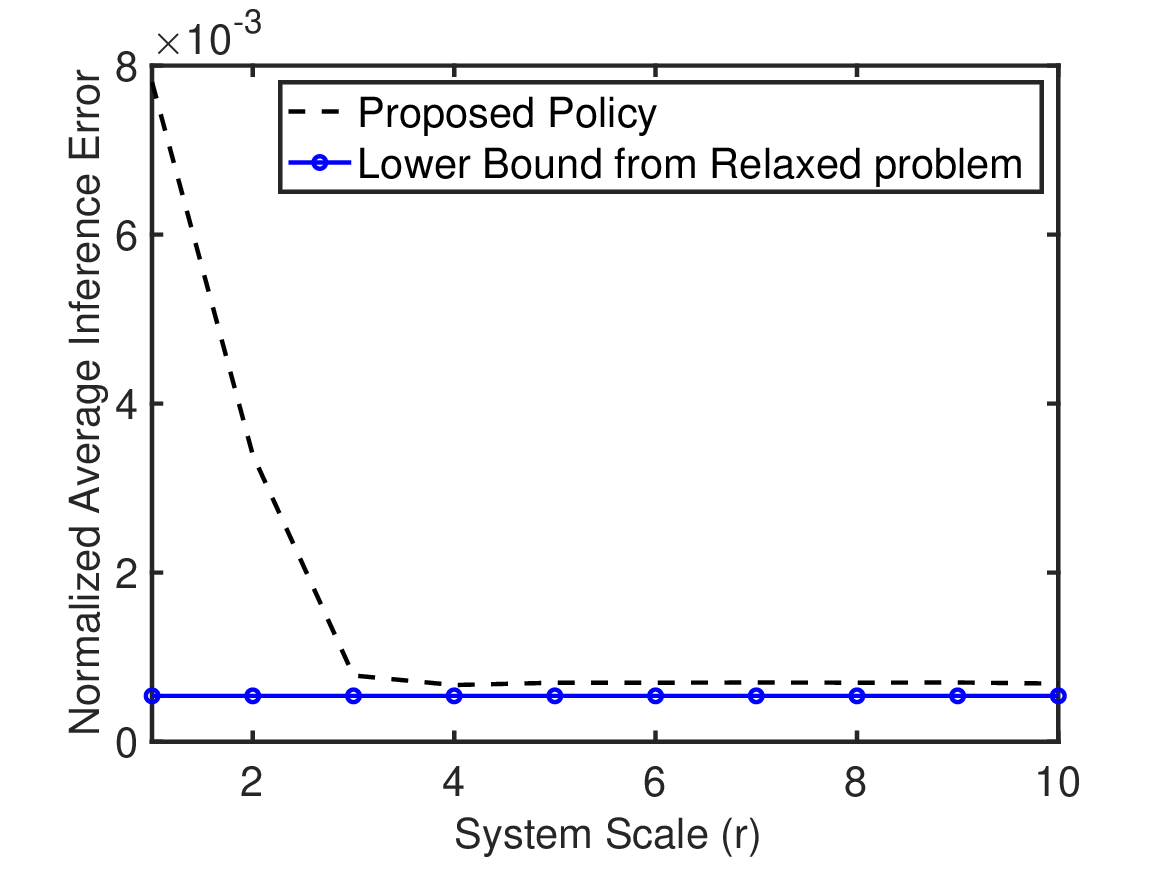}
\caption{\small Multiple Source Case: Time-averaged inference error vs. system scale $r$, where $M=3r$ and $N=10r$.  \label{fig:numericalmultiple1}
}
\vspace{-3mm}
\end{figure}

 \subsection{Multiple Source Scheduling Policies}
In this section, we evaluate the time-averaged inference error of the following three multiple source scheduling policies. 
\begin{itemize}
\item[1.] Maximum Age First (MAF), Generate-at-will, $l=1$: As the name suggests, this policy selects the sources with maximum AoI value at each time. Specifically, under this policy, $\min\{N, M\}$ sources with maximum AoI are selected. Moreover, the feature length and the feature position of the selected sources are $1$ and $0$, respectively. 

\item[2.] Maximum Age First (MAF), Generate-at-will, $l=B$:  This policy also selects the sources with maximum AoI values at each time, but with feature length $l=B$. Under this policy, $\min\{\lfloor\frac{N}{B}\rfloor, M\}$ sources with maximum AoI are selected, where $B$ is the buffer size of all sources, i.e., $B_j=B$ for all source $j$. Moreover, the feature position of the selected sources is $0$. 

\item[3.] Proposed Policy: The policy in Algorithm \ref{alg:multischeduling}.
\end{itemize}

The performance of three multiple source scheduling policies is illustrated in Fig. \ref{fig:numericalmultiple11}, where each source sends its observed CSI to the corresponding predictor. In this simulation, three types of sources are considered: (i) type 1 source with a velocity of $v_1=15~\text{m/s}$ and a CSI variance of $b_{1}=0.5$, (ii) type 2 sources with a velocity of $v_2=20~\text{m/s}$ and a CSI variance of $b_{2}=0.1$, and (iii) type 3 sources with a velocity of $v_3=25~\text{m/s}$ and a CSI variance of $b_{3}=1$.  

{Fig. \ref{fig:numericalmultiple11} illustrates the normalized average inference error (the total time-averaged inference error divided by the number of sources) plotted against the number of sources $M$, with $N=100$ and $B=10$. We can observe from Fig. \ref{fig:numericalmultiple11} that when the number of sources is less, the normalized average inference error of our proposed policy is $10^4$ times better than ``MAF, Generate-at-will, $l=1$.'' However, ``MAF, Generate-at-will, $l=B$" is close to the proposed policy. But, When number of sources is more than $400$, the normalized average inference error becomes $4$ times lower than that of the ``MAF, Generate-at-will, $l=B$" policy. As the number of sources increases, the normalized average inference error obtained by ``MAF, Generate-at-will, $l=1$" becomes close to the normalized average inference error of the proposed policy. 

Fig. \ref{fig:numericalmultiple1} compares the time-averaged inference error of the proposed policy and a lower bound from a relaxed problem. The lower bound is achieved by selecting feature length and feature position by using \eqref{optimallengtht} and \eqref{OptimalPositiont}, respectively under dual cost $\lambda=\lambda^*$ obtained from Algorithm \ref{alg:lagrange}. For this evaluation, we have taken step size $10^{-4}/(kr)$ at each iteration $k$ In Algorithm \ref{alg:lagrange}. In Fig. \ref{fig:numericalmultiple1}, we consider $N=10r$ channels and $M=3r$ sources, where $r$ represents the system scale. Observing Fig. \ref{fig:numericalmultiple1}, it becomes evident that our proposed policy converges towards the lower bound as the system scale $r$ increases. }

\section{Conclusion}
This paper studies a learning and communications co-design framework that jointly determines feature length and transmission scheduling for improving remote inference performance. In single sensor-predictor pair system, we propose two distinct optimal scheduling policies for (i) time-invariant feature length and (ii) time-variant feature length. These two scheduling policies lead to significant performance improvement compared to classical approaches such as periodic updating and zero-wait policies. Using the Lagrangian decomposition of a relaxed formulation, we propose a new algorithm for multiple sensor-predictor pairs. Simulation results show that the proposed algorithm is better than the maximum age-first policy.
\bibliographystyle{IEEEtran}
\bibliography{refshisher}

% \pagebreak
\begin{appendices}

\ifreport
\section{Experimental Setup for Two Machine Learning Experiments}\label{ExperimentalSetup}

In the first experiment: wireless channel state information (CSI) prediction, our objective is to infer the CSI of a source at time $t$ by observing a feature consisting of a sequence of stale and noisy CSIs. Specifically, we consider a Rayleigh fading-based CSI. The dataset is generated by using the Jakes model \cite{jakes1994microwave}. In the Jakes fading channel model, the CSI can be expressed as a Gaussian random process. Due to the joint Gaussian distribution of the target and feature random variables, the optimal inference error performance is achieved by a linear MMSE estimator. Hence, a linear regression algorithm is adopted in our simulation. Nonetheless, our study can be readily applied to other neural network-based predictors. 

In the second experiment: actuator state prediction, we employ a neural network based predictor. In this experiment, we use an OpenAI CartPole-v1 task \cite{brockman2016openai} to generate a dataset, where a DQN reinforcement learning algorithm \cite{mnih2015human} is utilized to control the force on a cart and keep the pole attached to the cart from falling over. Our goal is to predict the pole angle at time $t$ based on a sequence of stale information of cart velocity with length $l$. The predictor in this experiment is an LSTM neural network that consists of one input layer, one hidden layer with 64 LSTM cells, and a fully connected output layer. Additional experiments can be found in a recent study \cite{shisher2021age, ShisherMobihoc, Shisher2023Timely}, including (a) video prediction and (b) robot state prediction.

\section{Proof of Lemma \ref{lemma2}}\label{plemma2}
\textbf{Part (a):} Consider the sequence $\tilde X^{l}_{-\delta} = (\tilde V_{-\delta}, \tilde V_{-\delta-1}, \ldots, \tilde V_{-\delta-l+1})$. It can be demonstrated that for any $1 \leq l_1 \leq l_2$, the Markov chain $\tilde Y_0 \leftrightarrow \tilde X^{l_2}_{-\delta} \leftrightarrow \tilde X^{l_1}_{-\delta}$ holds. This is due to the fact that for $1 \leq l_1 < l_2$, the sequence $\tilde X^{l_2}_{-\delta}=(\tilde V_{-\delta}, \tilde V_{-\delta-1}, \ldots, \tilde V_{-\delta-l_1+1},\tilde V_{-\delta-l_1}, \ldots, V_{-\delta-l_2+1})$ includes $\tilde X^{l_1}_{-\delta}=(\tilde V_{-\delta}, \tilde V_{-\delta-1}, \ldots, \tilde V_{-\delta-l_1+1})$ as well as $(\tilde V_{-\delta-l_1}, \ldots, \tilde V_{-\delta-l_2+1})$. By applying the data processing inequality \cite[Lemma 12.1]{Dawid1998} for $L$-conditional entropy, we can deduce that
\begin{align}\label{plemma2_1}
H_L(\tilde Y_0 | \tilde X^{l_2}_{-\delta}) \leq H_L(\tilde Y_0 | \tilde X^{l_1}_{-\delta}).
\end{align}

\textbf{Part (b):} Assuming that \eqref{condition2} holds for all $l=1, 2, \ldots$ and $x \in \mathcal{V}^l$, and employing \cite[Lemma 3]{ShisherMobihoc}, \cite{Shisher2023Timely} yields 
\begin{align}\label{plemma2_2}
H_L(P_{Y_t|X^{l}_{t-\delta}}; P_{\tilde Y_0|\tilde X^{l}_{-\delta}}|P_{X^{l}_{t-\delta}}) = H_L(Y_t|X^{l}_{t-\delta}) + O(\beta).
\end{align}

Combining \eqref{plemma2_2} with \eqref{plemma2_1}, we deduce that
\begin{align}
&H_L(P_{Y_t|X^{l_2}_{t-\delta}}; P_{\tilde Y_0|\tilde X^{l_2}_{-\delta}}|P_{X^{l_2}_{t-\delta}})\nonumber\\
=& H_L(Y_t|X^{l_2}_{t-\delta}) + O(\beta)\nonumber\\
\leq& H_L(Y_t|X^{l_1}_{t-\delta}) + O(\beta)\nonumber\\
=& H_L(P_{Y_t|X^{l_1}_{t-\delta}}; P_{\tilde Y_0|\tilde X^{l_1}_{-\delta}}|P_{X^{l_1}_{t-\delta}}) + O(\beta) + O(\beta)\nonumber\\
=& H_L(P_{Y_t|X^{l_1}_{t-\delta}}; P_{\tilde Y_0|\tilde X^{l_1}_{-\delta}}|P_{X^{l_1}_{t-\delta}}) + O(\beta).
\end{align}
This concludes the proof. \qed

\section{Proof of Theorem \ref{theorem1}}\label{ptheorem1}
\textbf{Optimal Feature Length Determination:} To find the optimal feature length for the time-invariant scheduling problem \eqref{scheduling_problem1}, we undertake a two-step process:

\begin{enumerate}
    \item \textit{Calculation of $\bar p_{l}$:} Given a feature length $l$, we start by determining $\bar p_{l}$, defined as
    \begin{align}\label{decomposed}
        \bar p_l = \inf_{\pi \in \Pi_{l}} \limsup_{T \rightarrow \infty} \frac{1}{T} \mathbb{E}_{\pi} \left[ \sum_{t=0}^{T-1} \mathrm{err}_{\mathrm{inference}}(\Delta(t), l) \right],
    \end{align}
    where $\Pi_{l}$ represents the set of admissible policies for feature length $l$. This step quantifies the optimal objective value for each specific feature length.

    \item \textit{Optimal Feature Length and Objective Value:} Having obtained $\bar p_{l}$ for all relevant $l$, the optimal feature length $l^*$ can be determined by solving
    \begin{align}\label{opt}
        l^* = \argmin_{l \in \mathbb{Z}: 1 \leq l \leq B} \bar p_l,
    \end{align}
    where $B$ represents an upper bound on the feature length. Additionally, the optimal objective value is given by
    \begin{align}\label{optp}
        \bar p_{\mathrm{inv}} = \min_{l \in \mathbb{Z}: 1 \leq l \leq B} \bar p_{l}.
    \end{align}
    These steps collectively identify the most suitable feature length and its corresponding objective value.
\end{enumerate}

We aim to solve the problem \eqref{scheduling_problem1} by addressing the sub problems \eqref{decomposed}-\eqref{opt}. Let's begin by solving \eqref{decomposed} using \cite[Theorem 4.2]{ShisherMobihoc}, restated below for completeness.
\begin{theorem}\label{theoremAppendix}\cite[Theorem 4.2]{ShisherMobihoc}
If the transmission times $T_i(l)$'s are i.i.d. with a finite mean $\mathbb E[T_i(l)]$, then there exists an optimal solution $\pi_l^*=((S^*_1(l), b^*_1(l), l), (S^*_2(l), b^*_2(l), l), \ldots) \in \Pi_{l}$ to \eqref{decomposed} that satisfies:
\begin{itemize}
\item[(a)] The optimal feature position in $\pi^*_l$ is time-invariant, i.e., $b^*_1(l)=b^*_2(l)= \dots=b^*(l)$. The optimal feature position $b^*(l)$ in $\pi^*_l$ is given by
\begin{align}\label{optimal_buffer_length_lb}
b^*(l)=\argmin_{ 0 \leq b \leq B-l} \beta_{b, l},
\end{align}
where $\beta_{b, l}$ is the unique root of equation \eqref{bisection}.

\item[(b)] The optimal scheduling time $S^*_{i+1}(l)$ in $\pi^*_l$ is determined by 
\begin{align}\label{OptimalWaitingTimelb}
S^*_{i+1}(l) = \min_{t \in \mathbb Z}\big\{ t \geq S^*_i(l)+T_i(l): \gamma_{l}(\Delta_{b}(t), l) \geq \bar p_{l}\big\},
\end{align}
where $\Delta_{b}(t)=t-S^*_i(l)+b^*(l)$ is the AoI at time $t$. The optimal objective value $\bar p_{l}$ of \eqref{decomposed} is 
\begin{align}\label{optimal_objectivebl}
\bar p_{l}=\min_{ 0 \leq b \leq B-l} \beta_{b, l}.
\end{align}
\end{itemize}
\end{theorem}

Using Theorem \ref{theoremAppendix}, we obtain values of $\bar p_l$ for all $l = 1, 2, \ldots, B$. We can then determine $l^*$ and $\bar p_{\mathrm{inv}}$ using \eqref{opt} and \eqref{optp}, respectively. Substituting $l^*$ and $\bar p_{\mathrm{inv}}$ into the policy $\pi^*_{l^*}$ established in Theorem \ref{theoremAppendix}, we derive the optimal policy $\pi^*$, as asserted in Theorem \ref{theorem1}. This completes the proof. \qed

\section{Proof of Theorem \ref{theorem2}}\label{ptheorem2}

The infinite time-horizon average cost problem \eqref{scheduling_problem} can be cast as an average cost semi-Markov decision process (SMDP) \cite{puterman2014markov, bertsekasdynamic}. To describe the SMDP, we define decision times, action, state, state transition, and cost of the SMDP.

\subsubsection{Decision Times and Waiting Time} Each $i$-th feature delivery time $D_i = S_i + T_i(l_i)$ is considered a decision time. Let $Z_{i+1}$ denote the waiting time between the $i$-th feature delivery time $D_i$ and the $(i+1)$-th feature sending time $S_{i+1}$, given by:
\begin{align}\label{def:waitingtime}
    Z_{i+1} = S_{i+1} - D_i.
\end{align}
With $S_0 = 0$, we can express $S_{i+1} = \sum_{k=0}^{i} T_k(l_k) + Z_{k+1}$. Thus, given $(T_0, T_1, \ldots)$, we can uniquely determine $(S_1, S_2, \ldots)$ from $(Z_1, Z_2, \ldots)$. Consequently, a policy in $\Pi$ can be represented as $\pi = ((Z_1, b_1, l_1), (Z_2, b_2, l_2), \ldots)$, where at time $D_i$, $(Z_{i+1}, b_{i+1}, l_{i+1})$ represents the action.

\subsubsection{State and State Transition} At time $D_i$, the state is $(\Delta(D_i), d(D_i))$. The AoI process $\Delta(t)$ evolves as:
\begin{align}\label{AoIProcess}
    \Delta(t) =
    \begin{cases}
        T_i(l_i) + b_i, & \text{if}~t = D_{i}, \quad i = 0, 1, \ldots, \\
        \Delta(t-1) + 1, & \text{otherwise}.
    \end{cases}
\end{align}
The feature length $d(t)$ evolves as:
\begin{align}\label{FeatureLengthProcess}
    d(t) =
    \begin{cases}
        l_i, & \text{if}~D_i \leq t < D_{i+1}.
    \end{cases}
\end{align}
Hence, at the decision time $D_i$, the state value is $(\Delta(D_i), d(D_i)) = (T_i(l_i) + b_i, l_i)$.

\subsubsection{Cost} The expected time between two decision times, $D_i$ and $D_{i+1}$, is given by:
\begin{align}\label{interdeliverytime}
    \mathbb E[D_{i+1} - D_i] =\mathbb E[Z_{i+1} + T_{i+1}(l_{i+1})].
\end{align}
Given $\Delta(D_i) = \delta$ and $d(D_i) = d$, the expected cost within the interval $[D_i, D_{i+1})$ is:
\begin{align}
    &\mathbb{E}\left[\sum_{t=D_i}^{D_{i+1}-1} \mathrm{err}_{\mathrm{inference}}\bigg(\Delta(t), d(t)\bigg)\right]\nonumber\\
    =& \mathbb{E}\left[\sum_{k=0}^{Z_{i+1}+T_{i+1}(l_{i+1})-1} \mathrm{err}_{\mathrm{inference}}\bigg(\delta+k, d\bigg)\right].
\end{align}

\textbf{Solution via Dynamic Programming:} To solve the SMDP, we employ the dynamic programming method \cite{puterman2014markov, bertsekasdynamic}. Initially, we define a $\sigma$-field and a stopping time set for the state process $(\Delta(t), d(t))$.

Define $\sigma$-field 
\begin{align}\label{sigma-field}
\mathcal F_s^t= \sigma((\Delta(t+k), d(t+k)): k \in \{0, 1, \ldots, s\}),
\end{align} which is the set of events whose occurrence are determined by the realization of the process $\{(\Delta(t+k), d(t+k)) :  k \in \{0, 1, \ldots, s\}\}$ from time slot $t$ up to time slot $t+s$. Then, $\{ \mathcal F_k^t, k \in \{0, 1, \ldots\}\}$ is the filtration of the process $\{S(t+k) : k \in \{0, 1, \ldots\} \}$. We define $\mathfrak M$ as the set of all stopping times by 
\begin{align}\label{StoppingTimes}
\mathfrak M=\{ \nu \geq 0 : \{\nu=k\} \in \mathcal F_k^t, k \in \{0, 1, 2, \ldots\}\}.
\end{align}

Given $(\Delta(D_i), d(D_i)) = (\delta, d)$, the optimal action $(Z^*_{i+1}, l^*_{i+1}, b^*_{i+1})$ satisfies the following Bellman optimality equation for the SMDP \eqref{scheduling_problem}:
\begin{align}\label{optimalitymain}
    &h(\delta, d)  \nonumber\\&
    =\min_{\substack{Z \in \mathfrak{M} \\ l \in \mathbb{Z}: 1 \leq l \leq B \\ b \in \mathbb{Z}: 0 \leq b \leq B-l}} \bigg\{\mathbb{E}\left[\sum_{k=0}^{Z+T_{i+1}(l)-1} \mathrm{err}_{\mathrm{inference}}(\delta+k, d)\right] \nonumber\\
    &\quad \quad \quad -\mathbb E[Z+T_{i+1}(l)]~ \bar{p}_{\mathrm{opt}}+ {\red \mathbb{E}[h(T_{i+1}(l) + b, l)]}\bigg\} \nonumber\\
    = &\min_{\substack{Z \in \mathfrak{M} \\ l \in \mathbb{Z}: 1 \leq l \leq B \\ b \in \mathbb{Z}: 0 \leq b \leq B-l}} \bigg\{\mathbb{E}\left[\sum_{k=0}^{Z+T_{1}(l)-1} \mathrm{err}_{\mathrm{inference}}(\delta+k, d)\right] \nonumber\\
    &\quad \quad \quad -\mathbb E[Z+T_{1}(l)]~ \bar{p}_{\mathrm{opt}}+ {\red \mathbb{E}[h(T_{1}(l) + b, l)]}\bigg\},
\end{align}
where $\mathfrak{M}$ is the set of stopping times defined in \eqref{StoppingTimes}, and the last equality holds because $T_i(l)$'s are independent and identically distributed.

The Bellman optimality equation \eqref{optimalitymain} is complex due to the need to jointly optimize three parameters: feature length $l$, feature position $b$, and waiting time $Z$. If $Z_l(\delta, d)$ defined in \eqref{stoppingtimesoln} represents the optimal waiting time for a given feature length $l$, then equation \eqref{optimalitymain} can be simplified as follows:
\begin{align}
    &h(\delta, d)= \nonumber\\
    & \min_{\substack{l \in \mathbb{Z}\\ 1 \leq l \leq B}}\bigg\{\mathbb{E}\left[\sum_{k=0}^{Z_l( \delta, d)+T_{1}(l)-1} \bigg(\mathrm{err}_{\mathrm{inference}}(\delta+k, d) - \bar{p}_{\mathrm{opt}}\bigg)\right] \nonumber\\
     &~~~~~~~~~~~~~~~~+ \min_{\overset{b\in \mathbb{Z}}{0 \leq b \leq B-l}} {\red \mathbb{E}[h(T_{1}(l)+b, l)]}\bigg\},
\end{align}
which leads to \eqref{Bellman2} and \eqref{optimal_featurelength}.

Now, we need to prove that $Z_l(\delta, d)$ is the optimal waiting time for a given feature length $l$. This is true if $Z_l( \delta, d)$ is the optimal solution to the following optimization problem:
\begin{align}\label{optimality}
 \min_{Z \in \mathfrak M} \mathbb E\left[\sum_{k=0}^{Z+T_{1}(l)-1} \bigg(\mathrm{err}_{\mathrm{inference}}(\delta+k, d)- \bar p_{opt}\bigg)\right].
\end{align}

\textbf{Simplification of the Problem \eqref{optimality}} The problem \eqref{optimality} poses a challenge due to its nature as an optimal stopping time problem. However, we can simplify the problem by exploiting a property of the state transition. Let $\nu^* \in \mathfrak{M}$ represent the optimal stopping time that solves \eqref{optimality}. For any $k \leq \nu^*$, it holds that $\Delta(D_i+k) = \Delta(D_i) + k$ and $d(D_i+k) = d(D_i)$. Consequently, the set $\{(\Delta(D_i+k), d(D_i+k)): k=1,2, \ldots, \nu^*\}$ is entirely determined by the initial value $(\Delta(D_i), d(D_i))$. Additionally, for all $k \leq \nu^*$, the $\sigma$-field $\mathcal{F}_{k}^{D_i}$ can be simplified as $\mathcal{F}_k^{D_i} = \sigma((\Delta(D_i), d(D_i))$. Thus, any stopping time in $\mathfrak{M}$ corresponds to a deterministic time. As a result, problem \eqref{optimality} can be further simplified into the following integer optimization problem:
\begin{align}\label{determinsticwaiting}
    \min_{Z \in \{0, 1, \ldots\}} \mathbb{E}\left[\sum_{k=0}^{Z+T_{1}(l)-1} \bigg(\mathrm{err}_{\mathrm{inference}}(\delta+k, d) - \bar{p}_{\mathrm{opt}}\bigg)\right].
\end{align}
We aim to demonstrate that $Z_l(\delta, d)$ is the optimal solution for \eqref{determinsticwaiting}.

\textbf{Optimal Waiting Time Determination:} 
By utilizing \eqref{determinsticwaiting}, we can determine that waiting time $Z=0$ is optimal if the following inequality holds:
\begin{align}
    &\min_{Z \in \{1,2, \ldots\}} \mathbb{E}\left[\sum_{k=0}^{Z+T_{1}(l)-1} \bigg(\mathrm{err}_{\mathrm{inference}}(\delta+k, d) - \bar{p}_{\mathrm{opt}}\bigg)\right] \nonumber\\
    &\geq \mathbb{E}\left[\sum_{k=0}^{T_{1}(l)-1} \bigg(\mathrm{err}_{\mathrm{inference}}(\delta+k, d) - \bar{p}_{\mathrm{opt}}\bigg)\right].
\end{align}

In scenarios where $Z=0$ is not optimal, the optimal choice becomes $Z=1$ if the following condition is satisfied:
\begin{align}
    &\min_{Z \in \{2,3 \ldots\}} \mathbb{E}\left[\sum_{k=0}^{Z+T_{1}(l)-1} \bigg(\mathrm{err}_{\mathrm{inference}}(\delta+k, d) - \bar{p}_{\mathrm{opt}}\bigg)\right] \nonumber\\
    &\geq \mathbb{E}\left[\sum_{k=0}^{1+T_{1}(l)-1} \bigg(\mathrm{err}_{\mathrm{inference}}(\delta+k, d) - \bar{p}_{\mathrm{opt}}\bigg)\right].
\end{align}

Following a similar argument, if $Z=0, 1, \ldots, \tau-1$ are not optimal, then $Z=\tau$ becomes optimal when
\begin{align}\label{tauoptimality}
    &\min_{Z \in \{\tau+1,\tau+2 \ldots\}}\mathbb E\left[\sum_{k=0}^{Z+T_{1}(l)-1} \bigg(\mathrm{err}_{\mathrm{inference}}(\delta+k, d)- \bar{p}_{opt}\bigg)\right] \nonumber\\
    &\geq \mathbb E\left[\sum_{k=0}^{\tau+T_{1}(l)-1} \bigg(\mathrm{err}_{\mathrm{inference}}(\delta+k, d)- \bar{p}_{opt}\bigg)\right].
\end{align}

Hence, we deduce that the optimal waiting time is the least integer value $\tau$ that satisfies \eqref{tauoptimality}. This inequality can be equivalently expressed as  
\begin{align}\label{tauoptimality1}
    &\min_{j \in \{1,2 \ldots\}}\mathbb E\left[\sum_{k=0}^{j-1} \bigg(\mathrm{err}_{\mathrm{inference}}(\delta+\tau+j+T_1(l), d)- \bar p_{opt}\bigg)\right] \nonumber\\
    &\geq 0.
\end{align}

Similar to Lemma 7 in \cite{orneeTON2021}, the following lemma holds.
\begin{lemma}\label{fraction_programming}
The following inequality holds 
\begin{align}
    &\min_{j \in \{1,2 \ldots\}}\mathbb E\left[\sum_{k=0}^{j-1} \bigg(\mathrm{err}_{\mathrm{inference}}(\delta+\tau+j+T_1(l), d)- \bar p_{opt}\bigg)\right]\nonumber\\
    &\geq 0,
\end{align}
if and only if 
\begin{align}\label{Gittinssimilar}
    &\min_{j \in \{1, 2, \ldots\}} \frac{1}{j}\sum_{k=0}^{j-1} \mathbb E\left[\mathrm{err}_{\mathrm{inference}}\bigg(\delta+\tau+j+T_1(l), d\bigg)\right] \geq \bar p_{opt}. 
\end{align}
\end{lemma}
The left hand side of \eqref{Gittinssimilar} is exactly $\gamma_{l}(\delta+\tau, d)$ defined in \eqref{gittins}.

To conclude the proof, the optimal waiting time corresponds to the least integer value $\tau$ that satisfies $\gamma_{l}(\delta+\tau, d) \geq \bar p_{opt}$. This optimal waiting time leads to \eqref{stoppingtimesoln}.
\qed

 \else
\fi

\end{appendices}
\begin{IEEEbiography}
[{\includegraphics[width=1in,height=1.25in,clip,keepaspectratio]{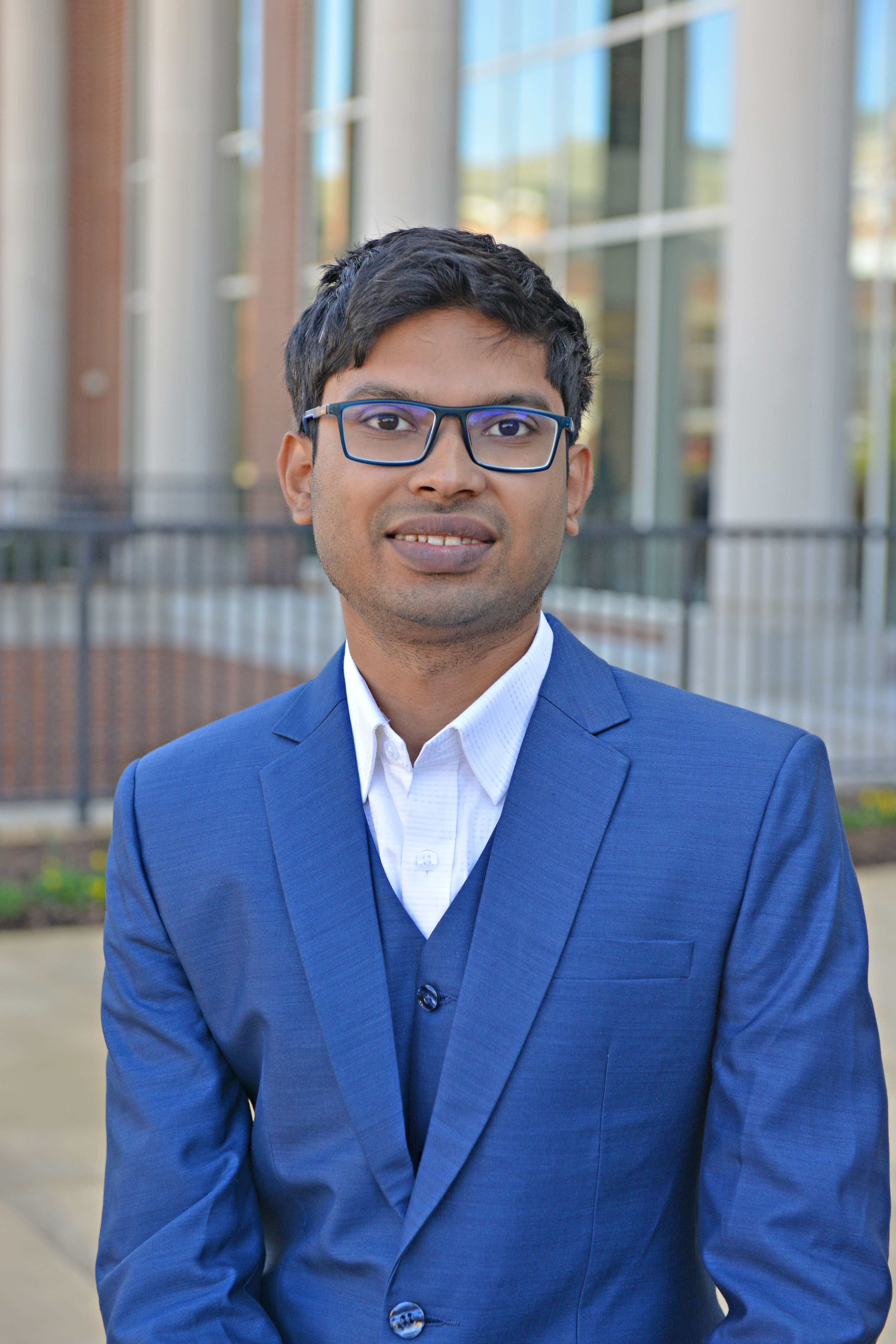}}]{Md Kamran Chowdhury Shisher}(Student Member, IEEE)
    received the B.Sc. degree in Electrical and Electronic Engineering from the Bangladesh University of Engineering and Technology in 2017. He completed his M.Sc. in Electrical Engineering at Auburn University in 2022. He is currently pursuing the Ph.D. degree with the Department of Electrical and Computer Engineering, Auburn University, Auburn, AL, USA. His research interests include Information Freshness, Wireless Networks, Semantic Communications, and Machine Learning.
\end{IEEEbiography}

\begin{IEEEbiography}[{\includegraphics[width=1in,height=1.25in,clip,keepaspectratio]{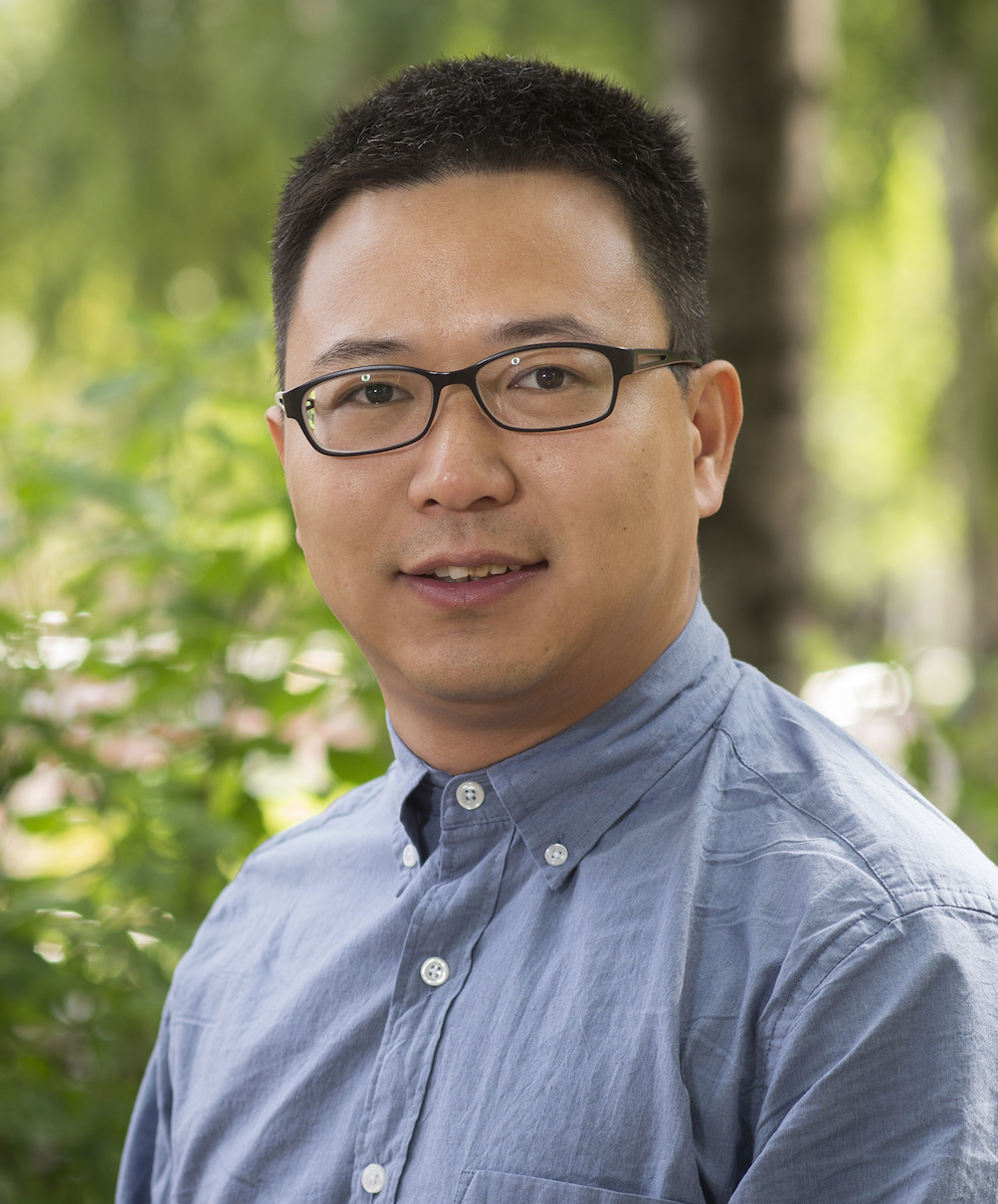}}]{Bo Ji}(Senior Member, IEEE)
received his B.E. and M.E. degrees in Information Science and Electronic Engineering from Zhejiang University, Hangzhou, China, in 2004 and 2006, respectively, and his Ph.D. degree in Electrical and Computer Engineering from The Ohio State University, Columbus, OH, USA, in 2012. Dr.~Ji is an Associate Professor of Computer Science and a College of Engineering Faculty Fellow at Virginia Tech, Blacksburg, VA, USA. Prior to joining Virginia Tech, he was an Associate/Assistant Professor in the Department of Computer and Information Sciences at Temple University from July 2014 to July 2020. He was also a Senior Member of the Technical Staff with AT\&T Labs, San Ramon, CA, from January 2013 to June 2014. His research interests are in the modeling, analysis, control, and optimization of computer and network systems, such as wired and wireless networks, large-scale IoT systems, high-performance computing systems and data centers, and cyber-physical systems. He has been the general co-chair of IEEE/IFIP WiOpt 2021 and the technical program co-chair of ACM MobiHoc 2023 and ITC 2021, and he has also served on the editorial boards of the IEEE/ACM Transactions on Networking, ACM SIGMETRICS Performance Evaluation Review, IEEE Transactions on Network Science and Engineering, IEEE Open Journal of the Communications Society (2019-2023), and IEEE Internet of Things Journal (2020-2022). Dr.~Ji is a senior member of the IEEE and the ACM. He was a recipient of the National Science Foundation (NSF) CAREER Award in 2017, the NSF CISE Research Initiation Initiative Award in 2017, the IEEE INFOCOM 2019 Best Paper Award, the IEEE/IFIP WiOpt 2022 Best Student Paper Award, and the IEEE TNSE Excellent Editor Award in 2021 and 2022.
\end{IEEEbiography}
\vspace{-1cm}
\begin{IEEEbiography}[{\includegraphics[width=1in,height=1.25in,clip,keepaspectratio]{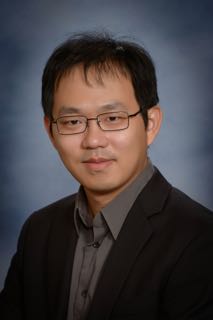}}]{I-Hong Hou}(Senior Member, IEEE) is an Associate Professor in the ECE Department of the Texas A\&M University. He received his Ph.D. from the Computer Science Department of the University of Illinois at Urbana-Champaign. His research interests include wireless networks, edge/cloud computing, and reinforcement learning. His work has received the Best Paper Award from ACM MobiHoc 2017 and ACM MobiHoc 2020, and Best Student Paper Award from WiOpt 2017. He has also received the C.W. Gear Outstanding Graduate Student Award from the University of Illinois at Urbana-Champaign, and the Silver Prize in the Asian Pacific Mathematics Olympiad.
\end{IEEEbiography}
\vspace{-1cm}
\begin{IEEEbiography}[{\includegraphics[width=1in,height=1.25in,clip,keepaspectratio]{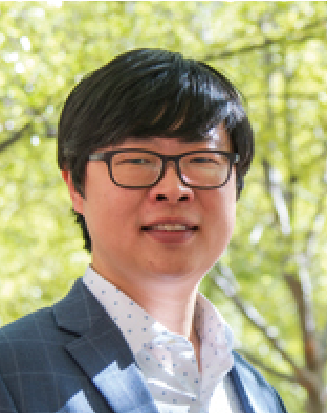}}]{Yin Sun}(Senior Member, IEEE) received his B.Eng. and Ph.D. degrees in Electronic Engineering from Tsinghua University, in 2006 and 2011, respectively. He was a Postdoctoral Scholar and Research Associate at the Ohio State University from 2011-2017 and an Assistant Professor in the Department of Electrical and Computer Engineering at Auburn University from 2017-2023. He is currently the Godbold Associate Professor in the Department of Electrical and Computer Engineering at Auburn University, Alabama. His research interests include Wireless Networks, Machine Learning, Semantic Communications, Age of Information, Information Theory, and Robotic Control. He is also interested in applying AI and Machine Learning techniques in Agricultural, Food, and Nutrition Sciences. 
He has been an Associate Editor of the \emph{IEEE Transactions on Network Science and Engineering}, an Editor of the \emph{Journal of Communications and Networks}, an Editor of the \emph{IEEE Transactions on Green Communications and Networking}, a Guest Editor of five special journal issues, and an Organizing Committee Member of several conferences. He founded the Age of Information (AoI) Workshop in 2018 and the Modeling and Optimization in Semantic Communications (MOSC) Workshop in 2023. His articles received the Best Student Paper Award of the IEEE/IFIP WiOpt 2013, Best Paper Award of the IEEE/IFIP WiOpt 2019, runner-up for the Best Paper Award of ACM MobiHoc 2020, and 2021 Journal of Communications and Networks (JCN) Best Paper Award. He co-authored a monograph \emph{Age of Information: A New Metric for Information Freshness}. He received the Auburn Author Award of 2020, the National Science Foundation (NSF) CAREER Award in 2023, and was named a Ginn Faculty Achievement Fellow in 2023. He is a Senior Member of the IEEE and a Member of the ACM. 
\end{IEEEbiography}
\end{document}